\documentclass{article}

\usepackage[linesnumbered,ruled,vlined]{algorithm2e}

\SetKwInput{KwInput}{Input}                
\SetKwInput{KwOutput}{Output}              

\usepackage{amsmath}
\usepackage{amsthm}
\usepackage{amssymb}
\usepackage{seqsplit}
\usepackage{color}
\usepackage{breakurl}
\usepackage{multirow}
\usepackage{indentfirst}
\usepackage{mdwlist}
\usepackage{booktabs}

\usepackage{graphicx}
\usepackage{rotating}

\usepackage{color}
\usepackage{subcaption}

\newcommand{\Not}{\text{\bf not}}
\newcommand{\tus}{\textunderscore}

\newtheorem{lemma}{Lemma}
\newtheorem{theorem}{Theorem}
\newtheorem{corollary}{Corollary}
\newtheorem{definition}{Definition}

\usepackage[version=4]{mhchem}

\usepackage[T1]{fontenc}

\usepackage{mathtools}
\DeclarePairedDelimiter{\ceil}{\lceil}{\rceil}

\usepackage{changepage}

\usepackage[utf8x]{inputenc}

\usepackage{textcomp,marvosym}

\usepackage{cite}

\usepackage{nameref,hyperref}

\usepackage[right]{lineno}

\usepackage{microtype}
\DisableLigatures[f]{encoding = *, family = * }

\usepackage[table]{xcolor}

\usepackage{array}

\newcolumntype{+}{!{\vrule width 2pt}}

\newlength\savedwidth

\begin{document}
\vspace*{0.2in}

\begin{flushleft}
{\Large
\textbf\newline{Selection on $X_1+X_2+\cdots + X_m$ with layer-ordered heaps} 
}
\newline
\\
Patrick Kreitzberg\textsuperscript{1},
Kyle Lucke\textsuperscript{2},
Oliver Serang\textsuperscript{2*}

\bigskip
\textbf{1} Department of Mathematics, University of Montana, Missoula, MT, USA
\textbf{2} Department of Computer Science, University of Montana, Missoula, MT, USA
\\
\bigskip

* oliver.serang@umontana.edu

\end{flushleft}

\begin{abstract}
  \noindent Selection on $X_1+X_2+\cdots + X_m$ is an important
  problem with many applications in areas such as max-convolution,
  max-product Bayesian inference, calculating most probable isotopes,
  and computing non-parametric test statistics, among others.
  Faster-than-na\"{i}ve approaches exist for $m=2$: Frederickson
  (1993) published the optimal algorithm with runtime $O(k)$ and
  Kaplan \emph{et al.} (2018) has since published a much simpler
  algorithm which makes use of Chazelle's soft heaps (2003). No fast
  methods exist for $m>2$. Johnson \& Mizoguchi (1978) introduced a
  method to compute the single $k^{th}$ value when $m>2$, but that
  method runs in $O(m\cdot n^{\lceil\frac{m}{2}\rceil} \log(n))$ time
  and is inefficient when $m \gg 1$ and
  $k \ll n^{\lceil\frac{m}{2}\rceil}$.

  In this paper, we introduce the first efficient methods, both in
  theory and practice, for problems with $m>2$. We introduce the
  ``layer-ordered heap,'' a simple special class of heap with which we
  produce a new, fast selection algorithm on the Cartesian
  product. Using this new algorithm to perform $k$-selection on the
  Cartesian product of $m$ arrays of length $n$ has runtime
  $\in o(k\cdot m)$. We also provide implementations of the algorithms
  proposed and evaluate their performance in practice.
\end{abstract}

\section*{Introduction}
Sorting all values in $A+B$, where $A$ and $B$ are arrays of length
$n$ and $A+B$ is the Cartesian product of these arrays under the $+$
operator, is nontrivial. In fact, there is no known approach faster
than naively computing and sorting them which takes
$O(n^2 \log(n^2))=O(n^2 \log(n))$\cite{bremner:necklaces}; however,
Fredman showed that $O(n^2)$ comparisons are
sufficient\cite{fredman:good}, though there is no known $O(n^2)$
algorithm. In 1993, Frederickson published the first optimal
$k$-selection algorithm on $A+B$ with runtime
$O(k)$\cite{frederickson:optimal}. In 2018, Kaplan \emph{et al.}
described another optimal method for $k$-selection on $A+B$, this time
in terms of soft heaps\cite{kaplan:selection}\cite{chazelle:soft}.

\subsection*{Soft heaps}
The soft heap was invented by Chazelle to be used in a deterministic
minimum spanning tree algorithm with the best known time
complexity\cite{chazelle:minimum}. The soft heap is a priority queue
similar to a binomial heap; however, the soft heap is able to break
the $\Omega(n\log(n))$ information theoretic runtime required by other
priority queues. This is done by allowing items in the heap to become
``corrupted'' by changing its key from the original value. With
corruption, popping from a soft heap $k$ times does not guarantee to
retrieve the best $k$ values. The soft heap has a tunable parameter,
$\epsilon \in \left(0, \frac{1}{2}\right)$, to control the amount of
corruption in the heap. In Chazelle's original implementation, all
soft heap operations have amortized constant runtime with the
exception of insertion which is $O(\log(\frac{1}{\epsilon}))$. The
number of corrupt items is bounded by $ \epsilon\cdot I$, where $I$ is
the number of insertions performed so far. Soft heaps perform the
standard insert, delete, and findmin operations.

In 2018, Kaplan \emph{et al.} introduced a new version of the soft
heap which is easier to implement and shifted the
$O(\log(\frac{1}{\epsilon}))$ runtime from the insert to the delete-min
operation. In this manuscript, we use an implementation of Kaplan
\emph{et al.}'s soft heap written in {\tt C++}.

\subsection*{Kaplan {et al.}'s Soft-Select algorithm}
Kaplan {et al.}'s optimal selection on $A+B$ is performed by using
{\tt Soft-Select}, an algorithm which is able to select the top $k$
items from a binary heap in $O(k)$.  This method works by first binary
heapifying $A$ and $B$, each in $O(n)$ time. Heapification orders $A$
so that $A_i \leq A_{2 i}$ and $A_i \leq A_{2 i + 1}$ (note that
1-indexing is throughout this manuscript), similar for $B$. This
heapification means that the child of some index $i$ may not appear in
the smallest $k$ values of $A+B$ if $i$ does not, since
$A_i + B_j \leq A_{2*i} + B_j, A_{2*i+1} + B_j$ for any $i,j$. Kaplan
\emph{et al.} gives a method by which we enumerate every value in
$A+B$ without duplicates where children are only considered after
their parents in the binary heap ordering, starting with the minimum
value $A_1 + B_1$ (equation~\ref{eqn:kaplan-proposal}). All currently
considered values (and their pairwise $(i,j)$ indices) are stored in a
soft heap.  When index $(i,j)$ is popped from the soft heap as a
candidate for the next lowest value, the following indices are
inserted into the soft heap as new candidates:
\begin{equation}
\label{eqn:kaplan-proposal}
\begin{cases}
  \{(2i,1), (2i+1,1), (i,2), (i,3)\}, & j=1\\
  \{(i,2j), (i,2j+1)\}, & j>1.\\
\end{cases}
\end{equation}
When corruption occurs, the corrupted items are appended to a list
$S$. The minimal $k$ values in the heap are obtained by doing a
one-dimensional selection on the union of the $k$ values popped and
$S$ median-of-medians\cite{blum:time} or soft heap-based
one-dimensional selection\cite{chazelle:soft}. Since the number of
corrupted items is $\leq \epsilon I$ and $I \leq 4k$ the soft heap
selection is linear in $k$.

In 1978, Johnson and Mizoguchi \cite{johnson:selecting} extended the
problem to selecting the $k^{th}$ element in
$X_1 + X_2 + \cdots + X_m$ and did so with runtime
$ O(m\cdot n^{\lceil\frac{m}{2}\rceil} \log(n))$; however, there has
not been significant work done on the problem since. Johnson and
Mizoguchi's method is the fastest known when
$k > n^{\lceil\frac{m}{2}\rceil}\log(n)$ but is significantly worse if
$k\ll n^{\lceil\frac{m}{2}\rceil}\log(n)$ and only returns a single
value, not the top $k$.

Selection on $X_1 + X_2 + \cdots + X_m$ is important for
max-convolution\cite{bussieck:fast} and max-product Bayesian
inference\cite{serang:fast,pfeuffer:bounded}. Computing the $k$ best
quotes on a supply chain for a business, when there is a prior on the
outcome (such as components from different companies not working
together) this becomes solving the top values of a probabilistic
linear Diophantine equation \cite{kreitzberg:toward} and thus becomes
a selection problem. Finding the most probable isotopes of a compound
such as hemoglobin, $C_{2952}H_{4664}O_{832}N_{812}S_8Fe_4$, may be
done by solving $C +H + O+N +S+Fe$, where $C$ would be the most
probable isotope combinations of 2,952 carbon molecules, $H$ would be
the most probable isotope combinations of 4,664 hydrogen molecules,
and so on \cite{kreitzberg:fast}.

\subsection*{Overview of the manuscript} In this paper, we construct
generalizations for selection on $X_1+X_2+\cdots + X_m$ when
$m \gg 1$, with and without soft heaps.  First, we develop a method
which is a direct generalization of Kaplan \emph{et al.}'s selection
on $A+B$ method we call SoftTensor. Second, we develop a method which
constructs a balanced binary tree where each interior node performs
pairwise selection on $A+B$ using {\tt Soft-Select} , we call this
method SoftTree. Third, we develop a method which sorts the input
arrays then builds an $m$-dimensional tensor which iteratively
retrieves the next smallest value until $k$ values are retrieved. This
method will output the top $k$ values in sorted order, we call this
method SortTensor. Fourth, we develop a method which constructs a
balanced binary tree using pairwise selection on $A+B$ where $A$ and
$B$ are sorted and iteratively retrieve the next smallest value. This
method also reports the top $k$ in sorted order, we call this method
SortTree. Fifth, we develop FastSoftTree, a method which constructs a
balanced binary tree of pairwise selection nodes which use a stricter
form of a binary heap to perform selection with runtime $o(k\cdot
m)$. Methods which use a soft heap use Kaplan {et al.}'s soft heap due
to the ease of implementation and the constant time insertions. The
theoretical runtime and space usage of each method is shown in
Table~\ref{table:runtimes-and-space}. All algorithms are implemented
in {\tt C++} and are freely available.

\begin{table}
  \begin{tabular}{r|ll}
    Method       & Runtime & Space\\
    \hline
    SoftTensor   & $O(k\cdot m^2)$ & $O(k\cdot m^2)$ \\
    SoftTree     & $O(k\cdot m)$ & $O(k\cdot m)$ \\
    SortTensor   & $O(k\cdot m^2 + k\log(k\cdot m\cdot n)$ & $O(min(k\cdot m^2, n^{m-1}))$ \\
    SortTree     & $O(k\log(k)\log(m))$ & $O(k\log(k)\log(m))$ \\
    FastSoftTree & $O(k\cdot m^{\log_2(\alpha^2)})$ & $O(k\log(m))$ \\
  \end{tabular}
  \caption{All runtimes are after the data is loaded and stored and,
    where applicable, the tree is constructed which takes time and
    space $\Theta(m\cdot n)$. In FastSoftTree, $\alpha$ is a constant
    parameter which will be explained below. Importantly, FastSoftTree
    can achieve runtime $O(k\cdot m^{0.141\ldots}) \in o(k\cdot m)$ by
    setting $\alpha=1.05$. }\label{table:runtimes-and-space}
\end{table}   


\section*{Materials and methods}
\subsection*{SoftTensor: a direct $m$-dimensional generalization of Kaplan \emph{et al.} 2018} 
SoftTensor is a straight-forward modification to {\tt Soft-Select}
that performs selection on $X_1 + X_2 + \cdots + X_m$. The most
significant modifications are the choice of $\epsilon$ and the
proposal scheme which must insert $\leq 2m$ children every time an item
is extracted from the soft heap. Similar to Kaplan \emph{et al.}, the
input arrays are binary heapified.

Indices are generated in a way that is a generalization of the
proposal scheme from Kaplan \emph{et al'}. For a given index
$(i_1,i_2,\ldots, i_\ell)$ there are two cases. The first, and more
simple case, is when $i_\ell > 1$. In this case only two children are
created and inserted into the soft heap, assuming they both exist:
$(i_1,i_2,\ldots, 2i_\ell)$ and $(i_1,i_2,\ldots, 2i_\ell + 1)$. In
the second case, let $i_j$ be the right-most index such that $i_j>1$
and $j<\ell$. Then $2(1+\ell - j)$ children will be created, each of
which will advance one index $\geq j$ in heap order, assuming the
children exist: $(i_1,i_2,\ldots, 2i_j, 1,1,1,\ldots)$,
$(i_1,i_2,\ldots, 2i_j+1, 1,1,1,\ldots)$,
$(i_1,i_2,\ldots, i_j, 2,1,1,\ldots)$,
$(i_1,i_2,\ldots, i_j, 3,1,1,\ldots)$,
$(i_1,i_2,\ldots, i_j, 1,2,1,\ldots)\ldots$. From the viewpoint of an
$m$-dimensional tensor this is equivalent to each of the $m$ axes
advancing in heap order.

Given any index $(i_1,i_2,\ldots, i_\ell)$, with $i_j$ defined as
above, the unique parent who proposed it is
$(i_1,i_2,\ldots, \lfloor\frac{i_j}{2}\rfloor, \ldots, i_\ell)$. Once
popped, no index will be re-inserted and therefore will only propose
their children once. The index of the parent will always be smaller
than the child and will never drop below zero. Therefore, there is a
path from every index to $(0,0,\ldots,0)$.  A child has all the same
indices as a parent except one is advanced in heap order, causing the
parent to always be higher priority than the children. Since all
indices can be reached from the original insertion $(0,0,\ldots,0)$,
and each child is inferior to its parent, this method must insert the
correct indices.

SoftTensor now proceeds exactly as {\tt Soft-Select}, only differing
in the insertion scheme. At each pop operation, up to $2m$ children
will be inserted by summing the appropriate $X_i$ entries (computed
na\"{i}vely in either $O(m)$ time or by observing that it has only a
single term changed from a previously created index, enabling
computation in $O(1)$ time) and creating their tuple index. After $k$
pops, the final values are computed by doing a one-dimensional
$k$-selection on $O(k)$ values. The fact that Kaplan \emph{et al}'s
soft heap has constant time insertions is important here since the
number of insertions will be much greater than the number of pops.

There has to be some care in selecting $\epsilon$ for this
method. {\tt Soft-Select} will pop $k$ items from the soft heap while
inserting $I < 2m(k + C)$ items, where $C$ is the number of items in
the soft heap which become corrupt. By the definition of the soft
heap, $C \leq \epsilon\cdot I < 2\epsilon m(k+C).$ Solving for $C$ we
get $C < \frac{2\epsilon \cdot k\cdot m}{1-2\epsilon m}$ and so
$I < 2m\cdot k\left( \frac{2m\epsilon}{1-2m\epsilon} \right)$. Therefore, we
need $\epsilon < \frac{1}{2m}$. In practice, we set
$\epsilon = \frac{1}{3m}$. For SoftTensor, popping from the soft heap
actually takes $O(\log(m))$; however this is dwarfed by simply
calculating the indices of the children of the popped item.

In this scheme, the $m$-dimensional tuple indices are necessary,
because they will be used to generate the subsequent children tuple
indices. Because of the necessity of storing the tuple indices, each
pop operation will require $\leq 2m$ insertions, which will each cost
$O(m)$, and thus each pop operation will cost $O(m^2)$. It may be
possible to compress these $m$-tuples to only reflect the ``edits'' to
the tuple which inserted them; however, we will later introduce the
SoftTree method, which performs a superior variant of this idea.

\subsubsection*{Space}
The space usage is dominated by the number of terms (which each
include an $m$-dimensional index) in the soft heap. The number of
index tuples in the soft heap after $k$ pop operations are performed
will be $\leq 2k\cdot m$. Each index tuple has size $m$, so the space
usage will be $O(k\cdot m^2)$. Thus, the total space usage, including
the $n\cdot m$ cost to store the arrays, is
$\in O(n\cdot m + k\cdot m^2)$.

\subsubsection*{Runtime}
The number of items in the one-dimensional selection, after the $k$
pops, will be $\leq k + C$. As described above,
$k + C < k + \frac{2\epsilon \cdot m\cdot k}{1-2\epsilon m} \in O(k)$.  Each pop
operation in the soft heap costs $O(\log(m))$ which each insertion
costs $O(m^2)$. Therefore, the runtime will be dominated by the cost
of insertions and will be $O(k\cdot m^2)$.

\subsubsection*{Implementation}
A pseudocode implementation of the selection function for SoftTensor
is provided in Algorithm~\ref{alg:softtensor-select}.

\begin{algorithm}[H]
\DontPrintSemicolon
  \KwInput{$k$, $m$ arrays of length $n$}
  \KwOutput{Top $k$ values}
  soft\textunderscore heap $\gets$ {\tt SoftHeap}($\epsilon=\frac{1}{3m}$)  \;
  soft\tus heap.insert($((1,1,\ldots,1), X_1[1] + X_2[1] + \cdots + X_m[1])$)\;
  corrupted\tus items $ \gets$ {\tt vector}\;
  results $ \gets$ {\tt vector}\;
  \For{$i \gets 1$ to $k$}
    {
      popped\tus item $\gets$ soft\tus heap.pop(corrupted\tus items)\;
      \If {popped\tus item.is\tus corrupt = false}
      {
        corrupted\tus items.append(popped\tus item)\;
      }

      \For{item : corrupted\tus items}
      {
        insert\tus children\tus into\tus soft\tus heap(item)\;
        results.append(item.key)\;
        }
      }
      \Return one\tus dimensional\tus select(results, k)\;
\caption{SoftTensor::select}\label{alg:softtensor-select}
\end{algorithm}

\subsection*{SoftTree: a soft heap-based balanced binary tree method}
The SoftTensor method suffers when $m\gg 1$, mainly due to the
$O(m^2)$ required during insertions. SoftTree is a method which
similarly uses soft heaps to perform selection on a Cartesian product
but without requiring calculating and storing indices of length $m$,
thus performing selection with runtime $o(k\cdot m^2)$.

SoftTree works by forming a balanced binary tree of nodes that perform
pairwise selection. The leaves take one of the $m$ input arrays and
when asked for $k$ values from a parent they simply perform a
one-dimensional selection on the input array. Being a balanced binary
tree with $m$ leaves, the tree is always finite. Nodes above the
leaves perform pairwise selection on $A+B$ using {\tt Soft-Select}
where $A$ and $B$ are the data generated from the left and right
children, respectively. It is nontrivial to stack these pairwise soft
heap selection nodes in this manner because they require their inputs
$A$ and $B$ to be binary heaps with random access. The output of these
pairwise soft heap selection nodes must therefore be made to be in
heap order with random access. The key is that any pairwise selection
only needs to keep track of length two indices, avoiding the $m^2$
cost of insertion in SoftTensor.

To perform selection on $X_1 + X_2 + \cdots + X_m$ the root node is
asked for its top $k$ values. The root node then asks its children to
generate values so that it may perform a pairwise selection, this
ripples down the tree to the leaves. If a child generates all of their
possible values anytime a parent needs to perform selection, this may
cause an exponential number of excess values generated.

\begin{lemma}
  \label{thm:a-plus-b-leq-k-plus-1}
  Let $a$ be the number of terms of $A$ required to produce the
  minimal $k$ values of the form $A_i+B_j$. Let $b$ be the number of
  terms of $B$ required to produce the minimal $k$ values of the form
  $A_i+B_j$. Then $a+b-1\leq k$.
\end{lemma}
\begin{proof}
Although the algorithm does not sort, here we refer to $A$ and $B$ in
ascending order: $A_i\leq A_{i+1}$ and $B_i\leq B_{i+1}$. Denote the
sorted tuples $(A_{i_s}+B_{j_s},i_s,j_s)\leq
(A_{i_{s+1}}+B_{j_{s+1}},i_{s+1},j_{s+1})$ sorted in ascending
order. Denote $S$ as the minimal $k$ terms: $S =
\{(A_{i_1}+B_{j_1},i_1,j_1), \ldots (A_{i_k}+B_{j_k},i_k,j_k)\}$.

By supposition, $\{i_1, i_2, \ldots i_k\}$ = $\{1, 2, \ldots a\}$ and
$\{j_1, j_2, \ldots j_k\}$ = $\{1, 2, \ldots b\}$.
$\forall i,j,~ (A_i+B_j,i,j)\geq (A_1+B_j,1,j),
(A_i+B_1,i,1)$. W.l.o.g., for any $i \in \{1, 2, \ldots a\}$,
$\exists j:~ (A_i+B_j,i,j)\in S$; therefore,
$(A_i+B_1,i,1)\leq (A_i+B_j,i,j)$ and thus $(A_i+B_1,i,1) \in S$.
$\left|\{(A_i+B_1,i,1) \mid i \in 1, 2, \ldots a \} \bigcup
  \{(A_1+B_j,1,j) \mid j \in 1, 2, \ldots b \}\right| =
\left|\{(A_i+B_1,i,1) \mid i \in 1, 2, \ldots a \}\right| +
\left|\{(A_1+B_j,1,j) \mid j \in 1, 2, \ldots b \}\right| -
\left|\{(A_1+B_1,1,1)\}\right| = a+b-1$.  Thus $S$ must contain these
$a+b-1$ terms, and so $a+b-1\leq k$.
\[
 \qedhere
\]
\end{proof}

By lemma~\ref{thm:a-plus-b-leq-k-plus-1}, the total number of values
required from $A$ and $B$ will together be $\leq k+1$. Thus, it is
sufficient to let $a=b=k$. This yields a simple, recursive algorithm
that proceeds in a pairwise manner: select the minimal $k$ terms from
the left child, select the minimal $k$ terms from the right child, and
then use only those values to select the minimal $k$ terms from their
Cartesian product.

\subsubsection*{Space}
Aside from leaves, every node in the tree will generate $k$ values
from its Cartesian sum, each in $O(k)$ time via optimal pairwise
selection on $A+B$. There are $m$ leaves, $\frac{m}{2}$ nodes on the
previous layer, \emph{etc.}, and thus $<2m$ total nodes in the tree
via geometric series. After the $n\cdot m$ cost of storing the input
data, the space requirement of this method is $O(k\cdot m)$.

\subsubsection*{Runtime}
For reasons similar to the space usage above (each pairwise $A+B$
selection node has a linear runtime in $k$), the runtime is
$O(k\cdot m)$ after loading the data.

\subsubsection*{Implementation}
A pseudocode implementation of the selection function for SoftTree is
provided in
Algorithms~\ref{alg:softtree-select} and \ref{alg:softselectnode-select}.

\begin{algorithm}[H]
\DontPrintSemicolon
  \KwInput{$k$, input\tus arrays}
  \KwOutput{Top $k$ values}
  root $\gets$ build\tus tree(input\tus arrays) \tcp*{{\tt SoftPairwiseSelectNode} object}
  \Return root.select($k$)
\caption{Select}\label{alg:softtree-select}
\end{algorithm}

\begin{algorithm}[H]
\DontPrintSemicolon
  \KwInput{$k$}
  \KwOutput{Top $k$ values of $A+B$}
  A = left\tus child.select(k) \;
  B = right\tus child.select(k) \;
  \Return SoftSelect(A,B,k) \tcp{Calls {\tt Soft-Select} from Kaplan \emph{et al.}}
\caption{SoftSelectNode::select}\label{alg:softselectnode-select}
\end{algorithm}

\subsection*{SortTensor: a direct $m$-dimensional generalization of the sorting-based $C=A+B$ method}
Forming a matrix sorted by both rows and columns can be a fast way to
get the minimum $k$ values from two vectors $A, B$. This can be
observed by sorting both $A$ and $B$ in ascending order, and then
sparsely building a matrix of $A_i+B_j$. If $A'$ and $B'$ represent
sorted vectors such that $A_1' \leq A_2' \leq \cdots$ and $B_1' \leq
B_2' \leq \cdots$, then the minimal value of $C$ is $A_1'+B_1'$. The
second smallest value in $C$ is $\min(A_1'+B_2', A_2'+B_1')$, and so
on. Whenever the next minimum value has been obtained its direct
neighbors can be inserted into the matrix, the collection of inserted
values which have not yet been popped form a ``hull.'' In practice, it
is efficient to use a min-heap to store the values in the hull, then
to get the next minimal value simply pop from the hull and then insert
the neighbors of the recently popped item.

The direct generalization of the $A+B$ method is straightforward:
Instead of a matrix, we have an $\mathbb{R}^m$ tensor, where the hull
is composed of a collection of $m$-dimensional indices. In each
iteration, we pop the minimal value from the hull and append that to
the result vector. Let this minimal value come from index
$(i_1,i_2,\ldots i_m)$. Now insert the $m$ values from the neighbors
of index $(i_1,i_2,\ldots i_m)$ into the heap holding the hull, if
they have not been inserted previously:
$(i_1+1,i_2,\ldots,i_m), (i_1,i_2+1,\ldots,i_m), \ldots
(i_1,i_2,\ldots,i_m+1)$. At most, each pop may insert $m-1$ indices
which are not already in the hull. As with the two-dimensional method,
it is possible to store not only the
$X_{1,i_1} + X_{2,i_2} + X_{3,i_3} + \cdots X_{m,i_m}$ in the heap,
but also store the index tuple from which it came.

Note that the selection of the minimal $k$ values in $X_1+X_2+\cdots +
X_m$ will be reported in sorted order.

\subsubsection*{Space}
The size of the hull, denoted using $h$, grows by $m-1$ values
(excluding the first iteration which adds $m$ values). Each of these
values will be accompanied by an $m$-dimensional index. Generalizing
the 2D case, the maximal size of $h$ is $\in O(n^{m-1})$, because it
will have the largest size as a diagonal hyperplane across all $m$
axes. Loading and storing the data is $\Theta(n\cdot m)$, and after
$k$ values are retrieved the hull will take storage space
$O(\min(k\cdot m^2, n^{m-1}))$.

\subsubsection*{Runtime}
Initially, the input vectors are heapified which costs $O(n\cdot
m)$. In each iteration, there are $O(m)$ objects inserted into the
hull's heap (one advancement per axis) and a single removal. Note that
each of the objects inserted includes an $m$-dimensional index. Using
a binomial heap (which has amortized $O(1)$ insert), the time to
insert each of these $O(m)$ objects is $O(1)$, but it takes $O(m)$
steps to construct each $m$-dimensional tuple index. Thus, the
insertion time per iteration is in $O(m^2)$; therefore, over $k$
iterations, it costs $O(k\cdot m^2)$.

In each iteration the hull grows in size by $O(m)$ values, so in the
first iteration the cost to pop from the hull is $O(\log(m))$, in the
second iteration the cost is $O(\log(2 m))$, \emph{etc}. Over all $k$
iterations the cost of popping is
$O(\log(m) + \log(2 m) + \dots + \log(k\cdot m)) = O(\log(m\cdot
(2\cdot m) \cdot (3\cdot m)\cdots (k\cdot m))) = O(\log(m^k \cdot k!))
= O(k\log(m) + k \log(k))$. 

In each iteration, the margin of at most one axis is advanced (see
invariant below). Advancing that margin requires popping the min value
from that input heap. This costs $O(\log(n))$ per iteration. In total,
the runtime after loading and heapifying the input arrays is
$O(n\cdot m + k\cdot m^2 + k\log(m) + k\log(k) + k\log(n)) = O(n\cdot
m + k\cdot m^2 + k\log(k\cdot n))$.

\subsubsection*{Implementation}
A pseudocode implementation of the selection function for SortTensor is provided
in
Algorithm~\ref{alg:sorttensor-select}.

\begin{algorithm}[H]
\DontPrintSemicolon
  \KwInput{$k$, $m$ arrays of length $n$}
  \KwOutput{Top $k$ values}
  result $\gets$ {\tt vector}\;
  fringe $\gets$ BonomialHeap\;
  fringe.insert($(1,1,\ldots,1), (X_{1,1} + X_{2,1} + \cdots + X_{m,1}))$ \;
  indices\tus in\tus fringe = set()\;
  \For{$i \gets 1$ to $k$}
    {
      popped\tus item $\gets$ fringe.pop()\;
      result.append(popped\tus item.key)\;
      \For{$j \gets 1$ to $m$}
      {
        $I \gets$ popped\tus item.index\;
        $I_j \gets I_j+1$\;
        \If {\text{$I$ not in indices\tus in\tus fringe}}
        {
          new\tus key = popped\tus item.key $+(X_{j,I_j} - X_{j,I_j-1})$\;
          fringe.insert(I, new\tus key)\;
          indices\tus in\tus fringe.insert(I)\;
        }
      }
    }
    \Return results\;
\caption{Select}\label{alg:sorttensor-select}
\end{algorithm}

\subsection*{SortTree: A sorting-based balanced binary tree}
In lemma~\ref{thm:a-plus-b-leq-k-plus-1} we saw that it is not
necessary to generate more than $k+1$ values total from the left and
right children combined; however, it is not trivial to discern how many
values we must generate from each child. One approach to resolve this
is to keep the $A$, $B$, and $A+B$ values in sorted order, generating
one at a time. Furthermore, it is not necessary for $A$ and $B$ to be
sorted vectors; instead, it is sufficient that $A$ and $B$ simply be
heap data structures, from which we can repeatedly request the next
smallest value.

The SortTree method runs similarly to the two-dimensional case: the
top-left corner of $A+B$ is computed via the minimal value of $A$ and
the minimal value of $B$. This inserts two values into the hull:
either the sum of the smallest value in $A$ and the second-smallest
value in $B$ or the sum of the second-smallest value in $A$ and the
smallest value in $B$. Neither the full, sorted contents of $A$ nor
the full, sorted contents of $B$ are needed.

We thus construct a balanced binary tree of these heap-like
structures. Except for the leaves (which are binary heaps of the $m$
arrays of length $n$), each heap-like structure is of the form $A+B$,
where $A$ and $B$ are heap-like structures
(figure~\ref{fig:sort-tree-cartoon}).

\begin{figure}
\centering
\includegraphics[width=.7\textwidth]{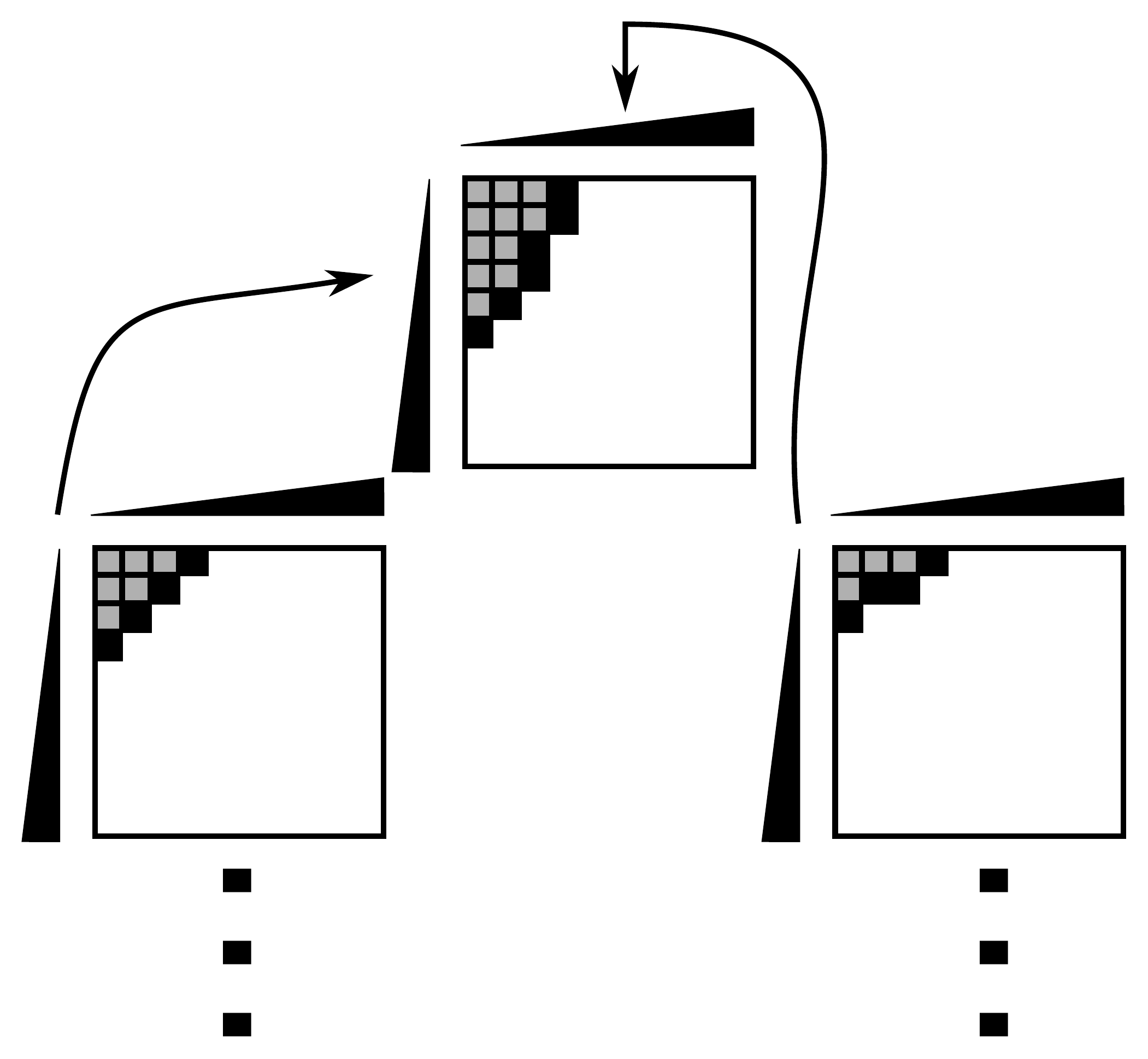}
 \caption{{\bf Illustration of the SortTree method.} Problems of the
   form $A+B$ are assembled into a balanced binary tree. The gray
   squares in each 2D matrix represent values which have already been
   popped from a heap of $A+B$ at the request of a parent node of
   $A+B$. When a value in a child is popped at the request of a
   parent, it advances the corresponding margin along one axis of the
   parent's matrix. The black squares are values in the hull, but which
   have not yet been popped. The child on the left has popped six
   values and currently has four values in its hull; the row axis of
   the parent has six values that have been realized thus far so
   $f_{vertical}=6$. The child on the right has popped four values and
   currently has four values in its hull; the column axis of the
   parent has four values that have been realized thus far so
   $f_{horizontal}=4$.
   \label{fig:sort-tree-cartoon}}
\end{figure}

For a node in the tree we call the set of values which are available
along each sorted axis the ``margin'' (these correspond to the gray
boxes in the child node corresponding to that axis in
Figure~\ref{fig:sort-tree-cartoon}). The size of the vertical margin,
$f_v$, is the largest sorted index requested from the left child (in
the figure, $f_v=6$). Similarly, the size of the horizontal margin,
$f_h$, is the largest sorted index requested from the right child (in
the figure, $f_h=4$).

The only time a parent node needs to request the next smallest value
from a child node is when the value of either $f_h$ or $f_v$ is
incremented. We have an invariant: after the first iteration, at any
node, the margin of at most one axis is incremented at a time. We
demonstrate by contradiction: To advance both margins simultaneously,
we have to pop cell $(f_v,f_h)$ from the hull's heap and select its
value as the next smallest. This adds indices $(f_v+1,f_h)$ and
$(f_v,f_h+1)$ to the hull, which would require popping from both the
left and right children (\emph{i.e.}, advancing both margins);
however, for $(f_v,f_h)$ to be in the hull, either $(f_v-1,f_h)$ or
$(f_v,f_h-1)$ must have already been popped. These pop operations
would have pushed either $(f_v,f_h)$ and $(f_v-1,f_h+1)$ or
$(f_v,f_h)$ and $(f_v+1,f_h-1)$ , respectively. Either of these would
have advanced the vertical or horizontal margins, making it impossible
for popping cell $(f_v,f_h)$ to advance both margins
simultaneously. The exception is when both $f_v-1$ is not a valid row
and and $f_h-1$ is not a valid column; in that case, we have
$f_v=f_h=1$, and thus both margins can be advanced only in the first
iteration.

From this invariant, we see that we will asymptotically propagate only
up to one child per iteration. The constant per-node cost of the first
iteration in all nodes (where we may visit both children) can be
amortized out during construction of the balanced binary tree: there
are $O(m)$ nodes in the tree. In a balanced binary tree, the longest
path from root to leaves that includes only one child per node is
simply a path from the root to a leaf, and thus propagates to
$\leq\log(m)+1$ nodes.

Note that the selection of the minimal $k$ values in $X_1+X_2+\cdots +
X_m$ will be reported in sorted order.

The $m$-dimensional indices from which these minimal $k$ values were
drawn can be optionally computed via depth-first search in $O(k\cdot
m)$ total.

\subsubsection*{Space}
The space needed to store the input vectors is $O(m\cdot n)$. From the
invariant above, in each iteration, only $O(\log(m))$ nodes will be
visited and they will add $O(1)$ new values into some hull; thus, the
space of all the hulls combined grows by $\log(m)$ per iteration. The
space usage after loading the data is $O(k\log(m))$. For
$k \leq \frac{m\cdot n}{\log(m)}$ the space usage is
$\in O(m\cdot n)$, which is the space required just to store the
original data.

If the $m$-dimensional indices are to be calculated and stored this
adds $O(k\cdot m)$ storage space to become $O(k\cdot m +
k\log(m))$. If $k \leq n$ then the space to store the indices is
$\in O(n\cdot m)$, which is the space to store the original data.

\subsubsection*{Runtime}
Initially, the input vectors are heapified which costs
$\in O(n\cdot m)$. Each hull starts with just the minimal possible
value. The first time a value is popped from the hull, two more values
will be placed into the hull. For every subsequent pop, only one new
value will be placed in to the hull, keeping the size of each hull
$\in O(k)$; therefore, the time to pop from any hull is
$\in O(\log(k))$.

Since both horizontal and vertical margins of a node cannot advance
simultaneously (except for the very first value, which is amortized
into the cost of tree construction), at most one child of each parent
will have to pop from its hull so there will be $O(\log(m))$ pops per
iteration. In the second iteration, the hulls will be of size 2 so
the time to pop from all hulls is $\in O(\log(2)\log(m))$, similarly,
for the third iteration the time to pop from all hulls is $\in
O(\log(3)\log(m))$, \emph{etc}. After $k$ iterations the total time
spent popping from hulls is $\in O(\log(k!)\log(m)) =
O(k\log(k)\log(m))$.

After loading the original data, the runtime of the SortTree method is
$O(n\cdot m + k\log(k)\log(m))$. If the $m$-dimensional indices are to
be calculated and stored (each length-$m$ tuple index is computed via
an in-order traversal of the tree, which costs $O(m)$ per iteration),
we add $O(k\cdot m)$ to the time to become
$O(k\log(k)\log(m) + k\cdot m)$.

If the distribution on the $X_i$'s lead to a fixed probability, $p <
1$, of advancing the margin in either child in the tree, then the
probability of descending to layer $\ell$ of the tree is $\in
O(p^{-\ell})$. Thus for such inputs, the expected runtime is a
geometric series. In that case, the $\log(m)$ coefficient disappears
and the total runtime without calculating the indices is $\in O(n\cdot
m + k\log(k))$ and total runtime with calculating the indices is $\in
O(n\cdot m + k\log(k) + k\cdot m)$.

\subsubsection*{Implementation}
A pseudocode implementation of the selection function for SortTree is
provided in Algorithms~\ref{alg:sorttree-select} and\ref{alg:pariwiseselectnode}.

\begin{algorithm}[H]
\DontPrintSemicolon
  \KwInput{$k$, input\tus arrays}
  \KwOutput{Top $k$ values}
  root $\gets$ build\tus tree(input\tus arrays) \tcp*{returns pairwise selection  object}
  results $\gets$ {\tt vector}\;
  \For{$i \gets 1$ to $k$}
  {
    popped\tus val, popped\tus index $\gets$ root.pop\tus min\tus value\tus and\tus index()\;  
    results.append(popped\tus val)\;
  }
  \Return results\;
\caption{Select}\label{alg:sorttree-select}
\end{algorithm}

\begin{algorithm}[H]
\DontPrintSemicolon
  \KwOutput{Min key and index in fringe}
  popped\tus val, popped\tus index $\gets$ fringe.pop()\;  
  insert\tus neighbors\tus into\tus fringe(popped\tus index)\; 
  generate\tus new\tus value\tus from\tus left\tus child\tus if\tus necessary()\;
  generate\tus new\tus value\tus from\tus right\tus child\tus if\tus necessary()\;
  insert\tus into\tus fringe(popped\tus index$_i+1$, popped\tus index$_j$)\;
  insert\tus into\tus fringe(popped\tus index$_i$, popped\tus index$_j+1$)\;
  generated\tus values.append((popped\tus val, popped\tus index))\;
  \Return (popped\tus val, popped\tus index)\;
\caption{PairwiseSelectNode::pop\tus min\tus value\tus and\tus index}\label{alg:pariwiseselectnode}
\end{algorithm}

\subsection*{FastSoftTree: an improved soft heap-based balanced binary tree method}
FastSoftTree works in a very similar way to SoftTree; however, where
every parent node in the SoftTree asks their children for exactly $k$
values each, a node in the FastSoftTree asks their children to
generate values only when needed by the parent to perform selection on
$A+B$. The $A+B$ selection still uses {\tt Soft-Select} to eschew the
$\Omega(n\log(n))$ bound required for sorting; however, it is this
on-the-fly generation of values from the children nodes that allows
FastSoftTree to achieve a runtime in $o(k\cdot m)$ (after tree
construction).

In the tree, the work cycle of a parent and its two children will
consist of three main steps. Let $C$ be the data produced by the
parent such that $C=A+B$. First, the parent will be asked to generate
values (possibly multiple times) so that its parent may perform a
selection. Second, the parent will ask its children to generate more
values to ensure that it has access to the values necessary for
selection on $A+B$. Third, the parent will then perform a slightly
modified {\tt Soft-Select} on $A+B$. The magic with FastSoftTree is
that all pairs of siblings will produce $\leq k+1$ values combined
without
sorting.  

\subsubsection*{Tree construction}
As with SoftTree, a finite, balanced binary tree is built which
combines two separate types of nodes. Each of the $m$ leaves in the
tree consists of a small wrapper around an individual input array. All
nodes above the leaves (including the root) are what we call $A+B$
nodes, as they use {\tt Soft-Select} to perform selection on the
Cartesian product of data generated by their children. Regardless of
whether a node is a leaf or an $A+B$ node, it has an output array
whose values are accessible only by their direct parent.

Unlike in the standard {\tt Soft-Select} algorithm, the input arrays
are not made into binary heaps, instead they are made into a slightly
stricter form of a heap structure: the layer-ordered heap (LOH)
(definition~\ref{thm:layer-ordered-heap}). However, LOHs are still
less informative than sorting as they can be constructed in $O(n)$ and
perhaps more surprisingly, are actually useful.

\subsubsection*{Layer-ordered heaps}
\begin{definition}
  \label{thm:layer-ordered-heap}
  A layer-ordered heap of rank $\alpha$ is a stricter form of heap
  structure which is segmented by layers with the $i$ layer denoted
  $L_i$.  An individual layer is defined to be all nodes which are the
  same number of generations removed from the root. The size of $L_i$
  is denoted $c_i$ and the ratio of the sizes of the layers
  $\frac{c_{i+1}}{c_i}$ tends to $\alpha$ as $i$ tends to infinity.
  Where in a standard heap, $\forall y \in children(x),~ x\leq y$, in
  a layer-ordered heap, $max(L_i) \leq min(L_{i+1})$ which we denote
  as $L_i \leq L_{i+1}$.
\end{definition}

A LOH may be constructed by performing partitions at a predetermined
set of pivot indices. The $i^{th}$ pivot index marks the first
value in layer $i$ and is calculated by
$p_i = max(p_i+c_i, \lceil \sum_{j=1}^i \alpha^j\rceil)$ with
$p_1 = 1, p_2=2$. It is necessary that $c_{i+1} \geq c_{i}$ so that
every node in layer $L_i$ has at least one child in $L_{i+1}$.

If $\alpha \leq 1$ then all layer sizes will be one, in which case
layer-ordering will sort the entire LOH. To avoid sorting, it is
enforced that $\alpha > 1$. If $\alpha<2$, indexing of children is
slightly more complicated than for a binary heap (where the children
of index $j$ are found at indices $2j,2j+1$).  The index of a node
relative to the start of the layer to which it belongs is its
``offset.'' The number of nodes in layer $i$ with one child will be
$d_{i,1} = 2c_i - c_{i+1}$, and the number of nodes with two children
$d_{i,2}$ will satisfy $d_{i,1} + d_{i,2} = c_{i}$, allowing us to
determine how many (and which) nodes in a particular layer have two
children (or one child). If a node in layer $i$ at offset $j$ has two
children, those children are at offsets $2j, 2j+1$ in the next
layer. Otherwise, the single child of a node at layer $i$ and at
offset $j$ will occur in the next layer at an offset of
$j - d_{i,2} + 2 d_{i,2} = j + d_{i,2}$. Rather than attempt to
compute closed forms, these values are cached by using the fact that
we will never visit layer $i$ before visiting every offset from every
layer to come before it.

A LOH may be generated online or it may be constructed \emph{in situ}
via an array. In the latter case, given $\alpha\in(1,2)$,
layer-ordered heapification (LOHification) of an array of length $n$
can be construction in
$O\left(n\log(\frac{1}{\alpha-1})\right) = O(n)$ since $\alpha$ is a
constant\cite{pennington:optimal}. Note that the choice of $\alpha$
may introduce a significant constant. This will be discussed below.

All nodes in the tree are considered to be ``layer-ordered heap
generators'', meaning the values their parents access are stored as a
LOH. For a leaf, this means the input array is LOHified upon
construction of the tree. Since this is done once, the cost is
amortized into the cost of loading the data. Nodes above the leaves
generate their LOHs one entire layer at a time. Since all nodes are
LOH generators, they inherit from the same base class. This means the
$A+B$ nodes do not know whether their child is another $A+B$ node or a
leaf node, they only know if the values in a layer of their children's
LOH are accessible or not.

During construction of the tree, each node will generate the first
layer of their respective LOH, which is a single value. This cost can
be amortized into the creation of the tree as it only requires the
parent to access the minimum element of its children which will be
$O(1)$ time after the leaves are initially LOHified.  When it is said
that a node generates a layer, this is done differently for a leaf
node versus an $A+B$ node. A leaf creates its entire LOH at
construction, so generating a new layer simply means to let the parent
access the next layer. For an $A+B$ node, generating a new layer
requires a selection on $A+B$.

\subsection*{The generation of an $A+B$ node's LOH}
When a parent is requested to generate a new layer, it must perform a
$k$-selection on $A+B$. This selection can only be done if the parent
has access to the requisite values from $A$ and $B$ (\emph{i.e.} the
children generated the necessary layers), but there is no way of
knowing which values are necessary to perform the selection \emph{a
  priori}.  

\subsubsection*{Concatenation selection}
Here, we show a means of finding a bounds on $a$ and $b$, the number
of values needed from $A$ and $B$, respectively, can be determined by
performing selection on the concatenation
$A|B = A_1, A_2, \ldots, B_1, B_2, \ldots$.

\begin{lemma}
  \label{thm:reduction-to-selection-on-concatenation}
  Let $a$ be the number of terms of $A$ required to produce the
  minimal $k$ values of the form $A_i+B_j$. Let $b$ be the number of
  terms of $B$ required to produce the minimal $k$ values of the form
  $A_i+B_j$. Bounds on $a$ and $b$ with $a+b\leq k+1$ can be found via
  a $(k+1)$-selection on $A|B$.
\end{lemma}
\begin{proof}
  Though $A$ and $B$ will not be sorted, for the ease of this proof
  assume they are indexed in ascending order: $A_i\leq A_{i+1}$ and
  $B_i\leq B_{i+1}$.
 
  Let $select(A+B,k)$ denote the minimal $k$ terms of the form
  $A_i+B_j$. It is trivial to show that for any finite scalar
  $\gamma$, selection on $A+B$ is bijective to selection on
  $A+B+\gamma$: $select(A+B,k) = select(A+B+\gamma,k)-\gamma$.

  Let $\gamma = -A_1-B_1$, then
  $select(A+B+\gamma,k) = select(A+B-(A_1+B_1),k) = select(A'+B',k)$
  where $A'_i = A_i-A_1$ and $B'_j = B_j-B_1$. Let every
  $A'_1, A'_2, \ldots A'_a$ and $B'_1, B'_2, \ldots B'_b$ be used at
  least once in some $A'_i+B'_j$ in $select(A'+B',k)$, where $a$ and
  $b$ are unknown. 

  Since all $A'_i,B'_j $ for $~i\in\{1,\ldots,a\},~j\in\{1,\ldots,b\}$
  are used in $select(A'+B',k)$, and for any $i>1,j>1$,
  $A'_i + B'_j \geq A'_i+B'_1, A'_1+B'_j$ it must be that
  $\{A'_1+B'_1, A'_2+B'_1,\ldots,A'_a+B'_1,
  A'_1+B'_2,A'_1+B'_3,\ldots,A'_1+B'_b\} \subseteq select(A'+B',k)$.
  Let
  $T = select(A'_1+B'_1, A'_2+B'_1,\ldots,A'_a+B'_1,
  A'_1+B'_2,A'_1+B'_3,\ldots,A'_1+B'_b,k)$.  By definition,
  $A'_1=B'_1=0$ and so
  $T = select(A'_1, A'_2,\ldots,A'_a,B'_2,B'_3,\ldots,B'_b,k)$. $A'_1$
  may be removed from the selection as it will always be in the select
  (since $A'_1=B'_1=0$), then
  $T=(A'_1,B'_1)|select(A'_2,A'_3,\ldots,A'_a, B'_2, \ldots,
  B'_b,k-1)$ which means
  $select(A'_1+B'_1, A'_2+B'_1,\ldots,A'_a+B'_1,
  A'_1+B'_2,A'_1+B'_3,\ldots,A'_1+B'_b,k) =
  (A'_1,B'_1)|select(A'_2,A'_3,\ldots,A'_a, B'_1, B'_2, \ldots,
  B'_b,k-1) = select(A'|B',k+1) = select(A|B + \gamma, k)$. This shows
  that the $(k+1)$-selection on $A'|B'$ returns all of the values
  from $A$ and $B$ which are possibly in $select(A+B,k)$.

  Thus bounds $s$ and $t$ can be found such that $a\leq s$ and
  $b\leq t$ with $s+t \leq k +1$ via a $(k+1)$-selection on $A'|B'$.
\[
 \qedhere
\]
\end{proof}

Now it is necessary to efficiently perform selection on $A|B$;
however, this is more difficult than it may seem because in a tree
where each node represents a problem of the form $A+B$, generating
every term in $A$ or $B$ will be combinatorial and thus could be
exponentially difficult. For this reason, we must perform the
selection on $A|B$ without generating every term in $A$ or $B$.

One way to perform selection on $A|B$ without generating many
unnecessary values would be to generate values in sorted order in a
manner reminiscent of the SortTree method; however, as we saw above,
sorting cannot achieve a runtime in $o(k\log(k))$. A heap order of $A$
or $B$ is less strict than computing $A$ or $B$ in sorted order;
however, the minimal $k$ values in a heap will not necessarily be
confined to any location in the heap. Thus, a heap is insufficient to
perform selection on $A|B$ while still generating few terms from both
$A$ and $B$ together. This is why LOHs are so critical to this
algorithm: they allow us to perform this concatenation select without
generating too many values
(theorem~\ref{thm:selection-by-generating-alpha-squared-k-terms}). Furthermore,
this method generalizes to perform online selection on $A|B$ so that
successively larger selections $k_1< k_2< k_3< \cdots $ can be
performed with runtime $O(\alpha^2 \cdot (k_1 + k_2 + k_3 + \cdots))$,
while asymptotically generating at most
$\alpha^2\cdot (k_1 + k_2 + k_3 + \cdots) = \alpha^2\cdot k$ terms of
both $A$ and $B$ combined (corollary~\ref{thm:online-selection}).

\begin{theorem}
  \label{thm:selection-by-generating-alpha-squared-k-terms}
  Consider two layer-ordered heaps $A$ and $B$ of rank $\alpha$, whose
  layers are generated dynamically an entire layer at a time (smallest
  layers first). Selection on the concatenation of $A,B$, can be
  performed by generating at most $\approx \alpha^2\cdot k$ values of
  $A$ and $B$ combined.
\end{theorem}
\begin{proof}
  Begin by generating the first layers, $A_1$ and $B_1$. If
  $\max(A_1)\leq \max(B_1)$, generate $A_2$; otherwise, generate
  $B_2$. Proceed iteratively in this manner: if
  $A_1, A_2, \ldots, A_x$ and $B_1, B_2, \ldots, B_y$ are extant
  layers, then generate layer $A_{x+1}$ if $\max(A_x)\leq \max(B_y)$;
  otherwise, generate layer $B_{y+1}$. Continue in this fashion until
  $u=c_1+c_2+\cdots +c_x + c_1+c_2+\cdots +c_y \geq k$.

  W.l.o.g., let $A_x$ be the layer whose generation resulted in
  $u\geq k$, breaking the iterative process described above. Before
  generating $A_x$, we must have $u<k$ or we would have terminated at
  the end of the previous iteration. Only one new layer, $A_x$, is
  generated. Thus we have
  $u'=c_1+c_2+\cdots +c_{x-1} + c_1+c_2+\cdots +c_y<k$ and
  $u=c_1+c_2+\cdots +c_x + c_1+c_2+\cdots +c_y\geq k$. The magnitude
  of $u$ compared to $u'$ is
  $\frac{u}{u'} = \frac{u' + c_x}{u'}=1+\frac{c_x}{u'}<1+\frac{c_x}{c_1+c_2+\cdots
    +c_{x-1}}$. Because of the iterated ceilings,
  $c_x \leq \ceil*{\ceil*{\alpha}\cdot \alpha}\ldots \leq
  ((\alpha+1)\cdot\alpha+1)\cdot\alpha+\cdots=\alpha^x+\alpha^{x-1}+\cdots+\alpha+1=\frac{\alpha^x-1}{\alpha-1}$. From
  this we have
  \[
  \frac{u}{u'}<\frac{\frac{\alpha^x-1}{\alpha-1}}{c_1+c_2+\cdots +c_{x-1}}<\frac{\frac{\alpha^x-1}{\alpha-1}}{\sum_{\ell=0}^{x-2}\alpha^\ell}=\frac{\frac{\alpha^x-1}{\alpha-1}}{\frac{\alpha^{x-1}-1}{\alpha-1}}=\frac{\alpha^x-1}{\alpha^{x-1}-1} \approx \alpha
  \]
  for large problems. $u'<k$, so $u<\frac{u}{u'}\cdot k$, and thus
  $u<\alpha\cdot k$.

  By the stopping criteria, there are at least $k$ values from
  $A_1, A_2, \ldots, A_x, B_1, B_2, \ldots, B_y$ that have been
  generated. Furthermore, by the layer-ordered property,
  $A_1, A_2, \ldots A_x, B_1, B_2, \ldots B_y \leq \max(\max(A_x),
  \max(B_y))$. Thus, $\max(\max(A_x), \max(B_y))$ is an upper bound on
  the partitioning value of the $k$-selection on $A|B$; however, we
  are not yet guaranteed to have the best $k$ values.  There are three
  cases that must be examined based on the relationship between
  $\max(A_x)$ and $\max(B_y)$.

  First, let's consider the case where $\max(A_x)<\max(B_y)$. Since
  $\max(B_y)$ is an upper bound on the partitioning value of the
  $k$-selection on $A|B$ and since the layer-ordered property dictates
  $\max(B_y)\leq B_{y+1} \leq B_{y+2}\leq \cdots$, no values in
  $B_{y+1}, B_{y+2}, \ldots$ will reach the $k$-selection of
  $A|B$. Thus, any values in the $k$-selection of $A|B$ that have not
  been generated may only be from $A_{x+1}, A_{x+2}, \ldots$. Because
  $A_x$ was generated after $B_y$, the state of the layers when $B_y$
  was generated were of the form
  \[
  \begin{array}{ccccccc}
    A_1 &\leq& A_2 &\leq& \cdots &\leq& A_s\\
    B_1 &\leq& B_2 &\leq& \cdots &\leq& B_{y-1},
  \end{array}
  \]
  where $s<x$. Since $B_y$ was generated before $A_{s+1}$,
  $\max(B_{y-1})<\max(A_s)$. Combined with the layer-ordered property,
  we get
  $B_1 \leq \cdots \leq B_{y-2} \leq B_{y-1} \leq \max(B_{y-1}) <
  \max(A_s) \leq A_{s+1} \leq \cdots \leq A_x \leq A_{x+1} \leq
  \cdots$. Thus, although the $k$-selection may require more values
  generated from $A_{x+1}, A_{x+2}, \ldots$, these values cannot
  displace any values from $B_1, B_2, \ldots B_{y-1}$; therefore, all
  values in $B_1, B_2, \ldots B_{y-1}$ are already in the
  $k$-selection. By the layer-ordered property, values generated from
  $A_{x+1}, A_{x+2}, \ldots$ may likewise not displace any values from
  $A_1, A_2, \ldots A_x$; therefore, new values from
  $A_{x+1}, A_{x+2}, \ldots$ may only displace in the $k$-selection
  extant values from $B_y$. Thus, at most $c_y$ additional values
  generated from $A_{x+1}, A_{x+2}, \ldots$ may be required.

  In the case where $\max(B_y)<\max(A_x)$, $\max(A_x)$ is an upper
  bound on the partitioning value used for $k$-selection on $A|B$;
  thus, no values of $A_{x+1}, A_{x+2}, \ldots$ need be considered by
  the $k$-selection. In this case, we need only consider additional
  layers generated in $B_{y+1}, B_{y+2}, \ldots$. We can exploit the
  fact that $\max(A_{x-1})\leq \max(B_y)$ (resulting in the generation
  of $A_x$ that halted the process above): by the layer-ordered
  property,
  $A_1\leq A_2\leq \cdots\leq A_{x-1}\leq \max(A_{x-1})\leq
  \max(B_y)\leq B_{y+1}\leq B_{y+2}\leq \cdots$, and thus new values
  generated from $B_{y+1}, B_{y+2}, \ldots$ may not displace from the
  $k$-selection values from $A_1, A_2, \ldots A_{x-1}$; therefore, all
  values in $A_1, A_2, \ldots A_{x-1}$ are already in the
  $k$-selection of $A|B$. Because of the layer-ordered property,
  $B_1 \leq B_2 \leq \cdots \leq B_y \leq B_{y+1}\leq \cdots$, values
  generated from $B_{y+1}, B_{y+2}, \ldots$ may not displace values
  from $B_1, B_2, \ldots, B_y$; these values are likewise in the
  $k$-selection of $A|B$. The only values that can be displaced by
  $B_{y+1}, B_{y+2}, \ldots$ are from $A_x$. Thus, at most $c_x$
  additional values generated from $B_{y+1}, B_{y+2}, \ldots$ may be
  required.

  In the final case, $\max(A_x)=\max(B_y)$. Like the two cases before,
  $\max(\max(A_x),\max(B_y))$ is an upper bound on the value threshold
  for $k$-selection on $A|B$; therefore, no larger values need be
  fetched to perform the $k$-selection. By the layer-ordered property,
  $A_x\leq A_{x+1}\leq \cdots$ and $B_y\leq B_{y+1}\leq \cdots$, and
  so any further layers of neither $A$ nor $B$ need be generated: by
  the termination of the iterative process above, we already have at
  least $k$ values $\leq \max(\max(A_x),\max(B_y))$, and so no values
  smaller than $\leq \max(\max(A_x),\max(B_y))$ can exist in
  $A_{x+1}, A_{x+2}, \ldots, B_{y+1}, B_{y+2}, \ldots$; therefore, we
  already have all values necessary to perform the $k$ selection on
  $A|B$.

  In the first two cases where $\max(A_x)\neq\max(B_y)$, additional
  layers may be generated with the total number of values
  $\leq c_y$ or $c_x$. W.l.o.g., these values will be generated
  sequentially from some layers $B_t, B_{t+1}, \ldots B_{t+\ell}$ such
  that $c_t + c_{t+1} + \cdots c_{t+\ell} \geq c_x$. By the same
  reasoning as the computation on $\frac{u}{u'}$ above, on large
  problems, this series will overestimate $c_x$ by at most a factor of
  $\alpha$.

  The total number of values generated in both layer-ordered heaps was
  $u<\alpha\cdot k$. In the worst-case scenario, $c_x$ is as large as
  possible and $y$ must be as small as possible (\emph{i.e.}, $y=1$);
  therefore, $u = c_1 ~+~ c_1+c_2+\cdots+c_{x-1}+c_x$. After adding
  additional layers, we have generated at most $u + \alpha\cdot c_x$
  where $\frac{u + \alpha\cdot c_x}{u'}= \frac{u + \alpha\cdot
    (u-u')}{u'} \approx \alpha + \alpha\cdot (\alpha-1) = \alpha^2$,
  for large problems.
  \[
  \qedhere
  \]
\end{proof}

\subsubsection*{Pairwise selection at an $A+B$ node}
The pairwise selection on $A+B$ scheme used in the SoftTree approach
requires that $A$ and $B$ be arranged in a binary heap where
$A_i \leq A_{2i}, A_{2i+1}$. In that scheme, selection of a term
$A_i+B_j$ may insert candidates which contain
$A_{2i},A_{2i+1},B_{2i},B_{2i+1}$; however, as described above, this
scheme may generate exponentially many values of $A$ and $B$. With a
LOH, we know that a child's child is the inferior of a child's
sibling, relieving us of this fear.

We use the same proposal scheme as {\tt Soft-Select} with two
modifications. First, when an item is popped, the children proposed
come from cached values as explained above, not from the proposal
scheme in Kaplan \emph{et al.} Second, after $A$ and $B$ are generated
for the current selection (via the selection on $A|B$), any indices in
$A$ or $B$ which have not been generated are not inserted. Instead,
they are placed into a purgatory list, whose contents will only be
considered after $A$ or $B$ have generated a new layer in a subsequent
selection (again, during the selection on $A|B$, which is responsible
for adding layers to $A$ and $B$ in anticipation of a selection on
$A+B$). We use three purgatory lists: one to consider when $A$ has
generated a new layer, one to consider when $B$ has generated a new
layer, and one to consider only after $A$ and $B$ have both generated
new layers. Importantly, from the layer ordering, we know that the
children of some value from layer $i$ must be in layer $i+1$. Thus,
after the relevant LOH has generated a new layer, it must be that
either any $(i,j)$ index pair in the purgatory list is now accessible
or that all layers in the LOH have been generated and an out-of-bounds
$(i,j)$ index pair is never to be used.

\subsubsection*{Online pairwise selection at an $A+B$ node}
Theorem~\ref{thm:selection-by-generating-alpha-squared-k-terms} shows
that a single selection on $A|B$ will generated $O(k)$ many values,
but in order to dynamically generate layers efficiently we may need to
be able to do successive $k_1 < k_2 < \cdots$ selections while
generating only $O(k_1 + k_2 + \cdots)$ values.

\begin{corollary}
  \label{thm:online-selection}
  Consider two layer-ordered heaps $A$ and $B$ of rank $\alpha$, whose
  layers are generated dynamically an entire layer at a time (smallest
  layers first). Successive selections of
  $k_1 < k_2 < k_3 < \cdots$ on the concatenation of $A,B$
  can be performed in an online manner by generating at most
  $\alpha^2\cdot (k_1 + k_2 + k_3 +\cdots)$ values of $A$ and $B$
  combined.
\end{corollary}
\begin{proof}
  Consider cases where $\max(A_x)\neq\max(B_y)$ and where the
  algorithm from
  theorem~\ref{thm:selection-by-generating-alpha-squared-k-terms} adds
  additional layers after $u\geq k$. W.l.o.g., let
  $\max(B_y) < \max(A_x)$, and thus layers may be added to $B$. The
  algorithm from
  theorem~\ref{thm:selection-by-generating-alpha-squared-k-terms}
  finishes adding layers to $B$ when these accumulated layers
  $B_t, B_{t+1}, \ldots B_{t+\ell}$ have total size exceeding that of
  $A_x$: $c_t + c_{t+1} + \cdots c_{t+\ell} \geq c_x$.

  Modify the selection algorithm from
  theorem~\ref{thm:selection-by-generating-alpha-squared-k-terms} so
  that it may also terminate when the most recent layer added to $B$,
  $B_t$, has $\max(B_t) \geq \max(A_x)$. This does not alter the
  correctness of the algorithm, because $\max(A_x)$ was an upper bound
  on the partitioning value for the $k_1$-selection on $A|B$;
  therefore, no layers beyond $B_t$ where
  $\max(A_x) \leq \max(B_t) \leq B_{t+1} \leq B_{t+2} \leq \cdots$ may
  be included in the $k_1$-selection.

  In this modified algorithm, any layers added to $B$ would have been
  added regardless in a subsequent $k_2$-selection (wherein each
  iteration generated a new layer for the LOH with a smaller $\max$
  generated thus far) if $k_2 > k_1$.
  \[
  \qedhere
  \]
\end{proof}

If the runtime to generate terms in the LOHs of $A$ and $B$ (both of
rank $\alpha$) is constant per term generated (\emph{i.e.}, linear
overall), then the runtime (including the constant from $\alpha$) to
generate $k$ terms will be $O\left(\alpha^2\cdot k\right)$
(corollary~\ref{thm:concatenation-selection-runtime}).

\begin{corollary}
  \label{thm:concatenation-selection-runtime}
  Given two processes that generate layers for layer-ordered heaps $A$
  and $B$ of rank $\alpha$ where the runtime to generate terms of the
  layer-ordered heaps is linear in the total number of terms
  generated, $k$-selection on $A|B$ can be run in
  $O\left(\alpha^2\cdot k\right)$.
\end{corollary}
\begin{proof}
  By theorem~\ref{thm:selection-by-generating-alpha-squared-k-terms},
  $\approx\alpha^2\cdot k$ terms are generated from $A$ and $B$ combined.
  $k+\alpha^2\cdot k$ is used to perform linear-time $k$-selection on
  the final results (\emph{e.g.}, using
  median-of-medians\cite{blum:time} or
  soft-select\cite{chazelle:soft}). Thus the runtime is
  $O(\alpha^2\cdot k + k + \alpha^2\cdot k) = O(\alpha^2\cdot k)$.
  Note, the LOH arrays can not be resized every time a new layer is
  generated without avoiding a $\frac{1}{\alpha-1}$ amortization
  constant. Instead, store each layer in a separate array and a list
  of pointers to the beginning of the layers.
  \[
  \qedhere
  \]
\end{proof}

What enables the soft heap to perform {\tt Soft-Select} in
$o(k\log(k))$ is the corruption of items in the soft heap. For
$\epsilon = 0.25$, the amount of corruption after a $k$-selection is
$\leq \epsilon \cdot I = 2\cdot k$\cite{kaplan:selection}. In order to
stop corruption from accumulating over successive selections, the soft
heap is destroyed and rebuilt after each layer is generated so that
the corruption is relative to the current size of the soft heap and
not the number of insertions over the previous $k1,k2,\ldots$
selections. 

At the end of {\tt Soft-Select}, a one-dimensional selection is
performed. In order to keep all necessary values for future layer
generation, the values not included in the output of the selection are
stored in a list to be included during the next one-dimensional
selection.

\subsubsection*{Termination of FastSoftTree}
When FastSoftTree is asked to retrieve the top $k$ values it does so
by asking the root node to generate layers. The root will then require
its children to generate layers of their own, and this will trickle
down the tree. The algorithm terminates when the root node has
generated enough layers such that the size of its LOH is $\geq
k$. Then a one-dimensional selection is performed to retrieve the top
$k$ values from the root's LOH.

\subsubsection*{Space}
At any internal node, time to fetch is linear with the number of $A+B$
values selected from that node. Since these values and nothing
substantially more are stored, the time and space are
comparable. Thus, using the argument for the runtime below, we have
space usage $\Theta(n\cdot m)$ to store the input arrays plus
$O(k\cdot \log(m))$ to store values generated during the selection.

\subsubsection*{Runtime}
By corollary~\ref{thm:concatenation-selection-runtime}, the runtime to
perform a $k$-selection on $A|B$ is $O(\alpha^2\cdot k)$. Once the
possibly used layers of $A$ and $B$ are generated by selection on
$A|B$, the modified {\tt Soft-Select} on $A+B$ has runtime $O(k)$.

Note that the total problem size at the root is $k$. In the next
layer, the total problem size will be asymptotically
$\leq \alpha^2\cdot k$. Since $\alpha>1$, the total runtime will be
dominated by the work done at the $A+B$ nodes whose children are
leaves. Thus, the runtime of propagation in the tree is
$O\left(k\cdot {(\alpha^2)}^{\log_2(m)}\right) = O\left(k\cdot
  m^{\log_2({\alpha^2})}\right)$ and the runtime to load and LOHify
the $m$ arrays is
$\Omega \left ( m\cdot \left( n\cdot \log(\frac{1}{\alpha-1}) +
    \frac{n\cdot \alpha\log(\alpha)}{\alpha-1}\right )\right)$
\cite{pennington:optimal}.

We now seek to find a good value of $\alpha$. This could be done by
computing $\frac{\partial}{\partial \alpha}$ and finding the $\alpha$
value that achieves a zero in the partial derivative and thus an
extremum, but Mathematica 11 could not solve the resulting equation
for $\alpha$. While we can not solve for the optimal $\alpha$
directly, we may examine some cases.

If $\alpha=1$, then the LOHs will be sorted, causing FastSoftTree to
have the runtime of SortTree. For
$\alpha \leq \left(1 + \frac{1}{n}\right)$,
$n\log(\frac{1}{\alpha-1}) \approx n\log(n)$, causing the LOHification
to be just as bad as sorting. Choosing $\alpha = \sqrt{2}$ may in the
worst case lead to doubling the problem size in each successive layer
(leading to an algorithm with propagation time in $\Omega(k\cdot m)$,
equivalent to or worse than SoftTree). Thus the optimal $\alpha$ must
be in $ \left(1 + \frac{1}{n}, \sqrt{2}\right)$ which gives the
overall runtime of the algorithm
$O\left (m\cdot n \cdot\log\left( \frac{1}{\alpha-1}\right)\right)$ to
load the data and construct the tree plus
$O(k\cdot m^{\log_2(\alpha^2)})$ for the $k$-selection.

Note that when $n \gg k$, the cost of layer-ordered heapification on
the $m$ input vectors of length $n$ could be lessened by first
performing $(k+m-1)$-selection on the concatenated $X_i$ vectors (by
na\"{i}vely concatenating them and using linear-time one-dimensional
selection $O(m\cdot n + k)=O(m\cdot n)$ time); by
lemma~\ref{thm:reduction-to-selection-on-concatenation}, this finds
only the values that may contribute to the final result. Because each
array $X_i$ must have at least $|X_i|\geq 1$, there are $k-1$ free
values that may be distributed between the arrays. To LOHify the
resulting, trimmed $X_i$ vectors will be linear in each and so will be
$O\left(\sum_{i=1}^m |X_i| \right) = O(k+m)$. If sorting is used to
LOHify the trimmed $X_i$ arrays, then the runtime is worst when
w.l.o.g. $|X_1|=|X_2|=\cdots=|X_{m-1}|=1, |X_m|=k-1$, because
comparison sort is superlinear. This has a runtime $O(m + k\log(k))$.

Regardless of which term dominates, the FastSoftTree can easily
perform selection with runtime $o(k\cdot m)$ by simply setting
$\alpha=1.05$: This results in runtime $O(k\cdot
m^{0.141\ldots})$. Thus, the FastSoftTree method has the best
theoretical runtime of all methods introduced here.

\subsubsection*{Implementation}
A pseudocode implementation of the selection function for FastSoftTree
is provided in
Algorithms~\ref{alg:compute-next-layer} and \ref{alg:concatenationselect}.

\begin{algorithm}[H]
\DontPrintSemicolon
  \KwInput{$k$}
  \KwOutput{Top $k$ values}

  \While {number\tus values\tus generated < k}
  {
    compute\tus next\tus layer\tus if\tus available()\;
  }
  \Return select(LOH, $k$)\;
\caption{select}
\end{algorithm}

\begin{algorithm}[H]
\DontPrintSemicolon
  \KwInput{$k$}
  \KwOutput{Top $k$ values}
  concatenation\tus select(k)\;
  insert\tus purgatories()\;
  SoftSelect(k)
  soft\tus heap.rebuild();
\caption{Compute next layer}\label{alg:compute-next-layer}
\end{algorithm}

\begin{algorithm}
  \DontPrintSemicolon
  \KwIn{A sequence of integers $\langle a_1, a_2, \ldots, a_n \rangle$}
  \KwOut{The index of first location with the same value as in a previous location in the sequence}
  \While{total\tus values\tus generated\tus so\tus far $< k$}{
    \If{\Not A.more\tus layers\tus are\tus available()}{
      B.compute\tus next\tus layer\tus if\tus available()\;
    }
    \ElseIf{\Not B.more\tus layers\tus are\tus available()} {
      A.compute\tus next\tus layer\tus if\tus available()\;
    }
    \Else{
      \If {$A.max\tus key\tus generated() \leq B.max\tus key\tus generated()$ } {
        A.compute\tus next\tus layer\tus if\tus available()\;
        total\tus values\tus generated\tus so\tus far $ += $ A.size\tus of\tus most\tus recent\tus layer()\;
      }
      \Else{
        B.compute\tus next\tus layer\tus if\tus available()\;
        total\tus values\tus generated\tus so\tus far $ += $ B.size\tus of\tus most\tus recent\tus layer()\;
      }
    }
  }

  \If{A.max\tus value = B.max\tus value} {
    \Return\;
  }

  \tcp{WLOG let A.max\tus value $>$ B.max\tus value}
  additional\tus values\tus from\tus A$ = 0$\;
  \While {additional\tus values\tus from\tus $A < $ B.size\tus of\tus last\tus layer\tus generated} {
    A.compute\tus next\tus layer\tus if\tus available()\;
    additional\tus values\tus from\tus A$ += $ A.size\tus of\tus last\tus layer\tus generated\;
    \If {A.max\tus value\tus generated $\geq$ B.max\tus value\tus generated } {
      \Return\;
    }
  }
\caption{{\sc ConcatenationSelect}}\label{alg:concatenationselect}
\label{algo:duplicate}
\end{algorithm}                 

\section*{Results}
\subsection*{Memory use and runtime}
All plots are log-log plots in order to compare the theoretical
runtime of each method (\emph{e.g.} the slope of a linear runtime
method will have a slope of one while a quadratic method will have a
slope of two). Figure~\ref{fig:memory} plots how memory usage grows
relative to $k$ and relative to $m$ when each $X_i$ is uniformly
distributed and exponentially distributed. Figure~\ref{fig:runtime}
plots how runtime grows relative to $k$ and relative to $m$ when each
$X_i$ is uniformly distributed and exponentially distributed. Memory
usage was measured using the Valgrind, specifically its Massif tool.

\begin{table}
  \begin{tabular}{c|cc}
    &Uniformly distributed $X_i$ & Exponentially distributed $X_i$ \\
    \hline
    \raisebox{5\normalbaselineskip}[0pt][0pt]{\rotatebox[origin=c]{90}{Vary k}} & \includegraphics[width=.5\linewidth]{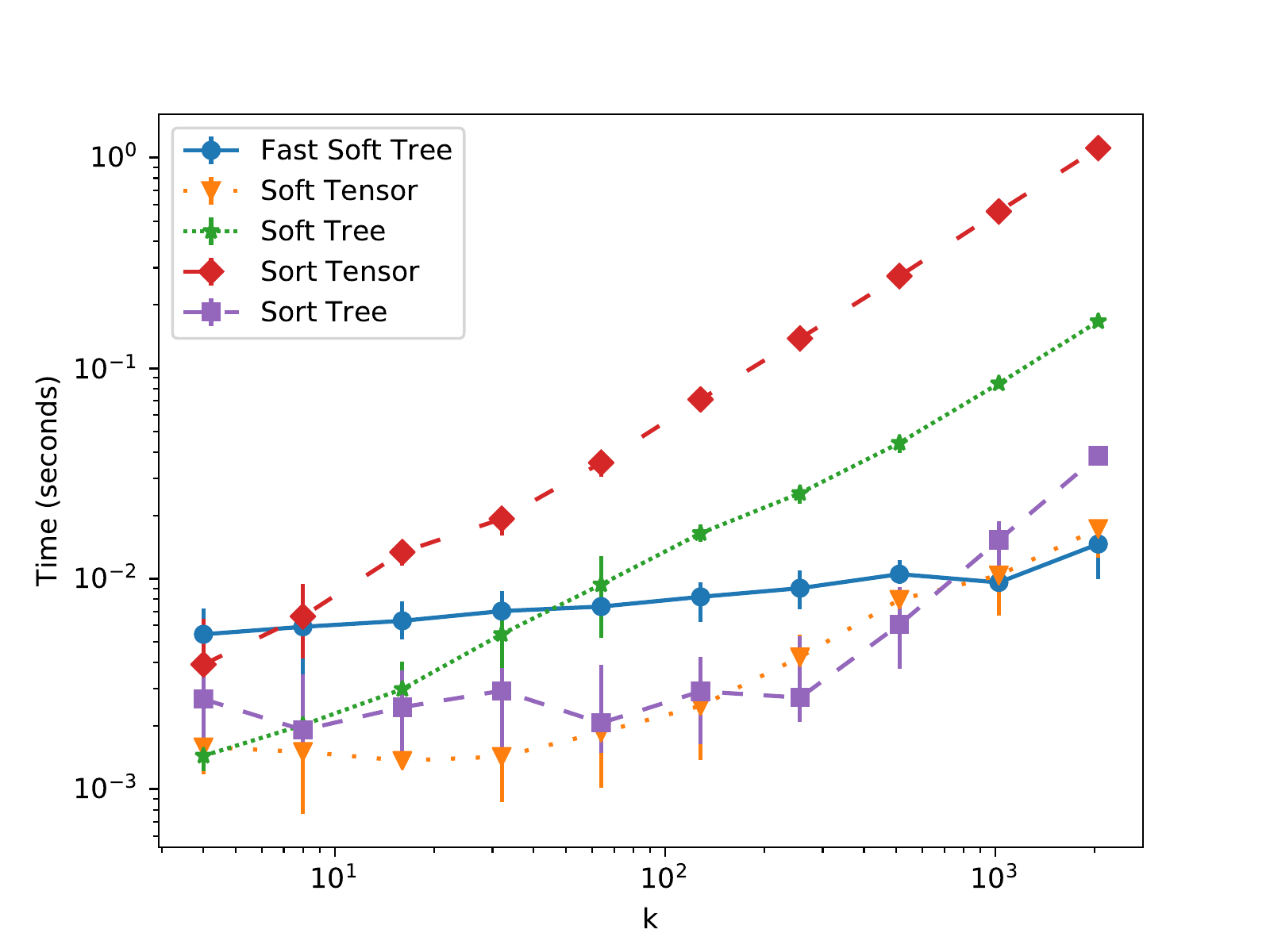} & \includegraphics[width=.5\linewidth]{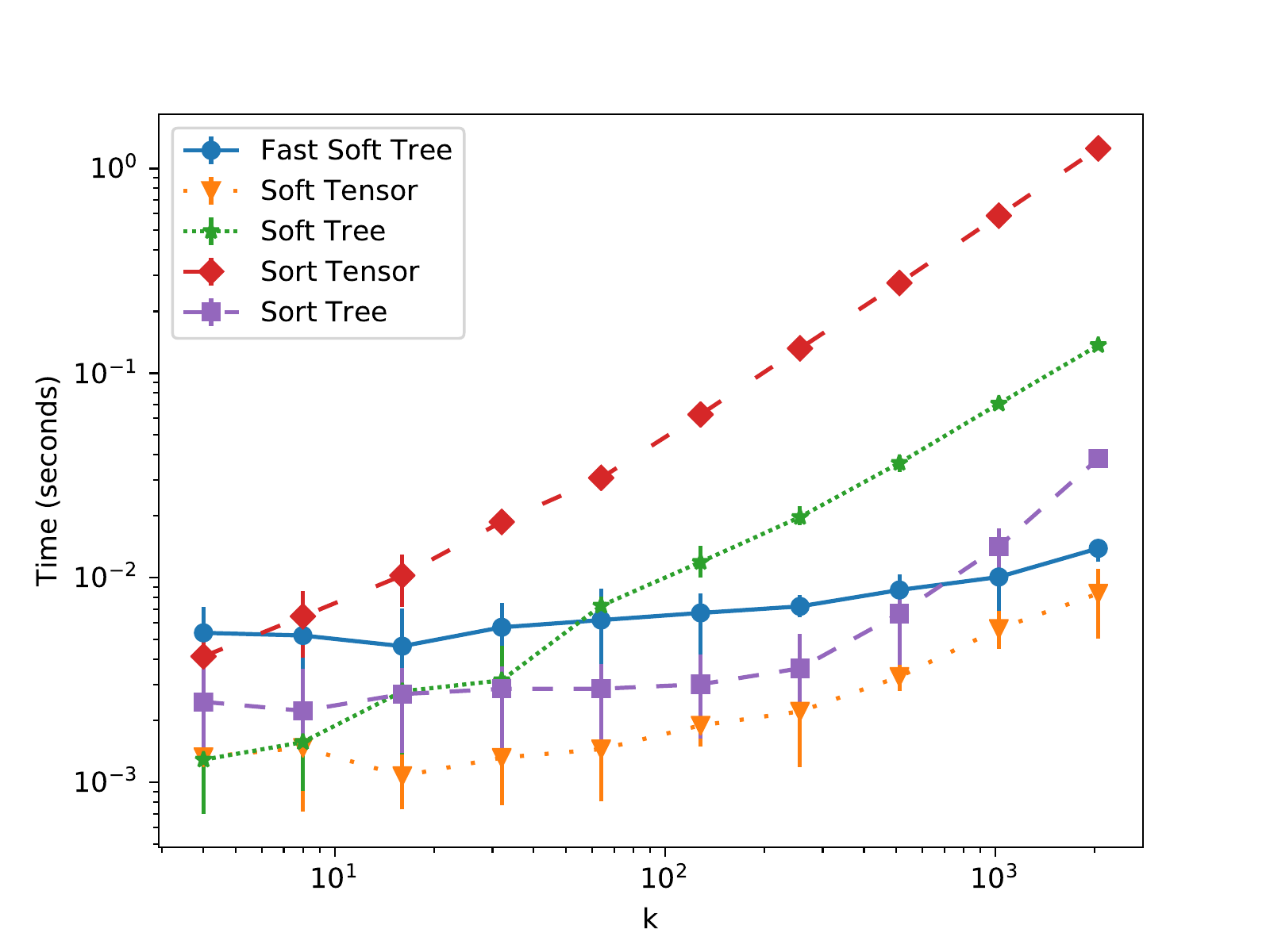}  \\
    \raisebox{5\normalbaselineskip}[0pt][0pt]{\rotatebox[origin=c]{90}{Vary m}} & \includegraphics[width=.5\linewidth]{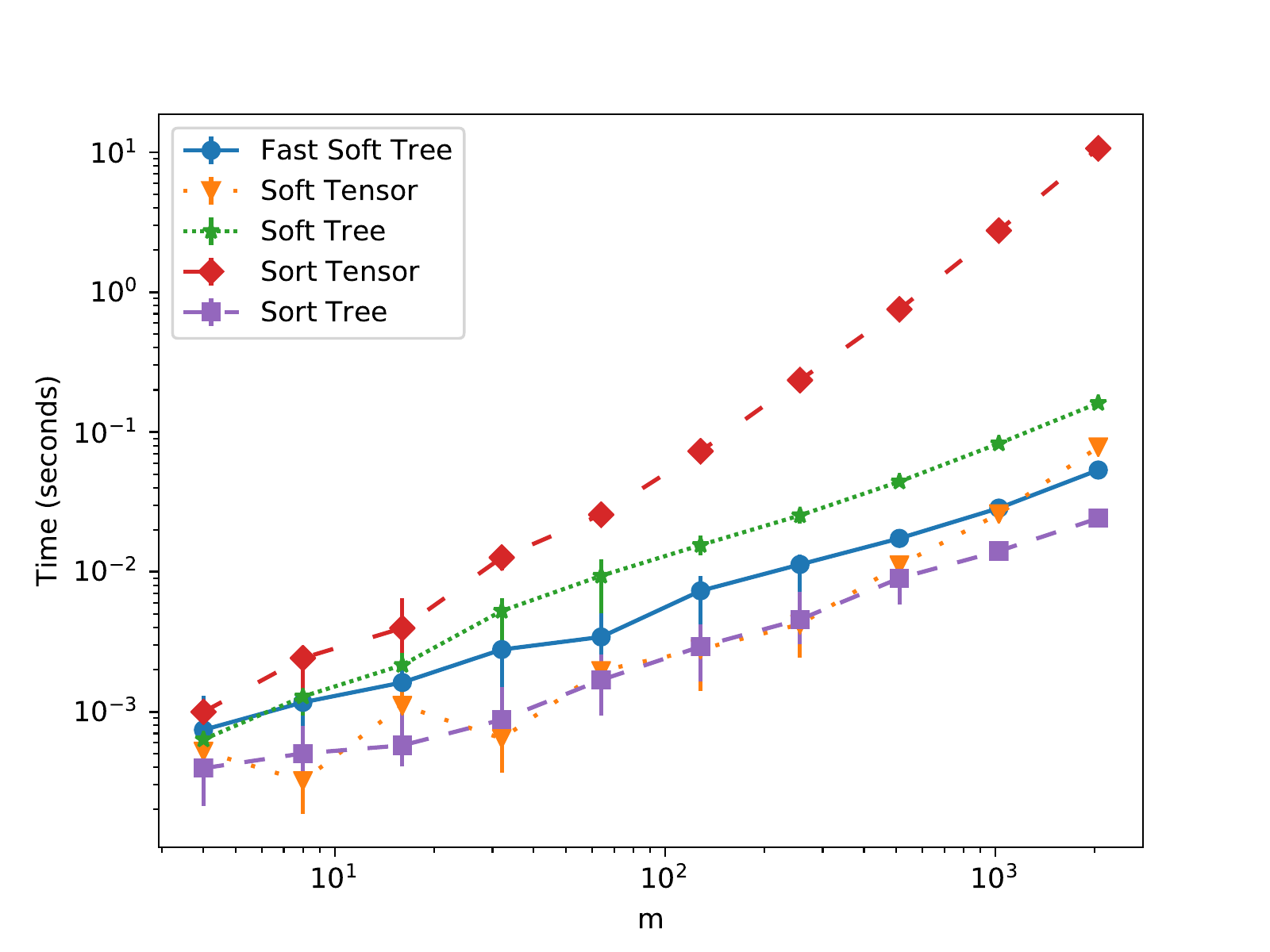} & \includegraphics[width=.5\linewidth]{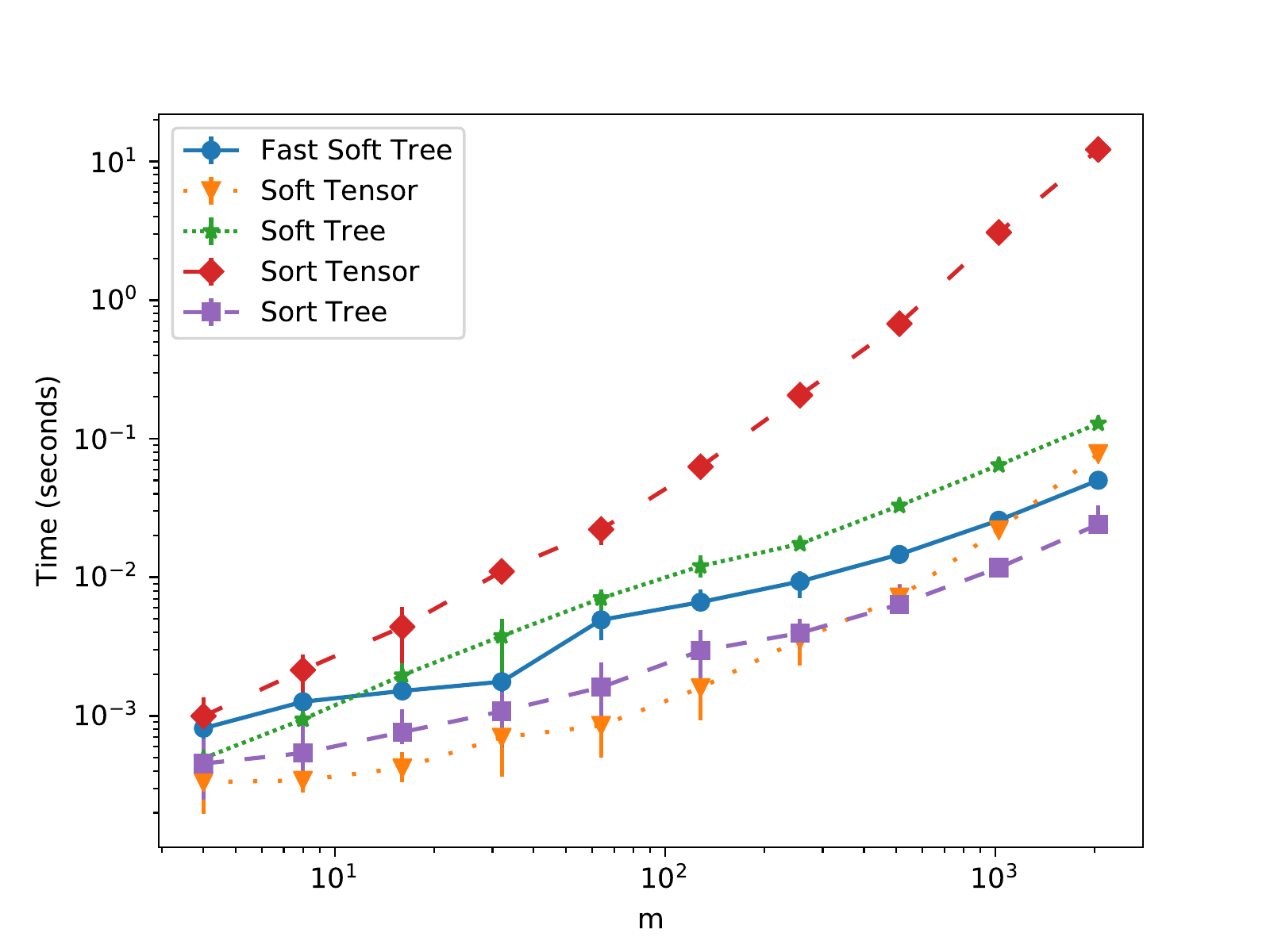}  \\
  \end{tabular}  \caption{{\bf Runtime.} Runtime for each method is plotted. When $k$
    is varied, it is $\in\{4, 8, 16, \ldots 256\}$ with $n=32$ and
    $m=128$. When $m$ is varied, it is $\in\{4, 8, 16, \ldots 256\}$
    with $n=32$ and $k=128$. We measure the average over twenty replicates
    with error bars showing the minimum and maximum. FastSoftTree was
    run with $\alpha=1.1$.
    \label{fig:runtime}}
\end{table}

\begin{table}
  \begin{tabular}{c|cc}
    &Uniformly distributed $X_i$ & Exponentially distributed $X_i$ \\
    \hline
    \raisebox{5\normalbaselineskip}[0pt][0pt]{\rotatebox[origin=c]{90}{Vary k}} & \includegraphics[width=.5\linewidth]{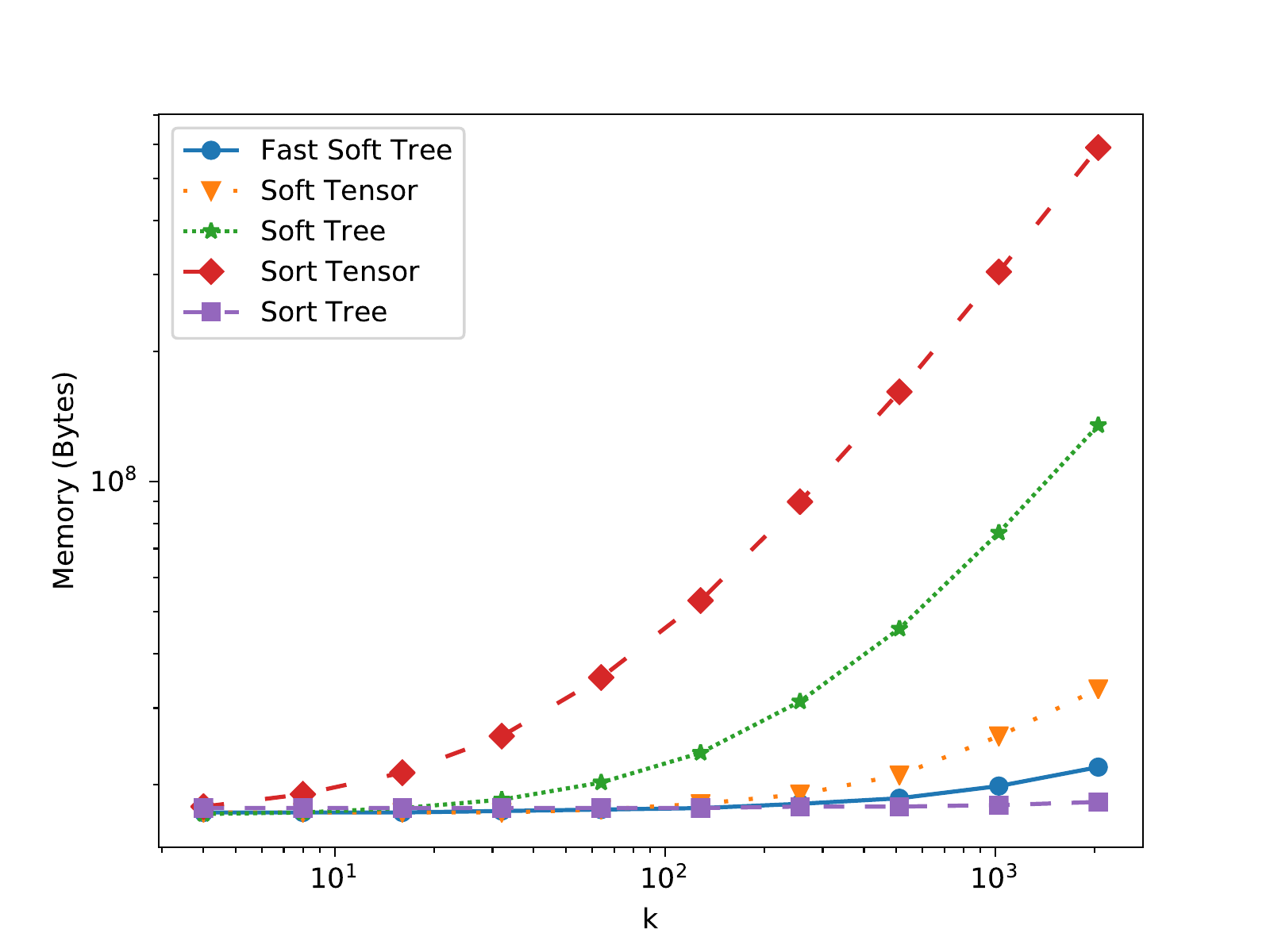} & \includegraphics[width=.5\linewidth]{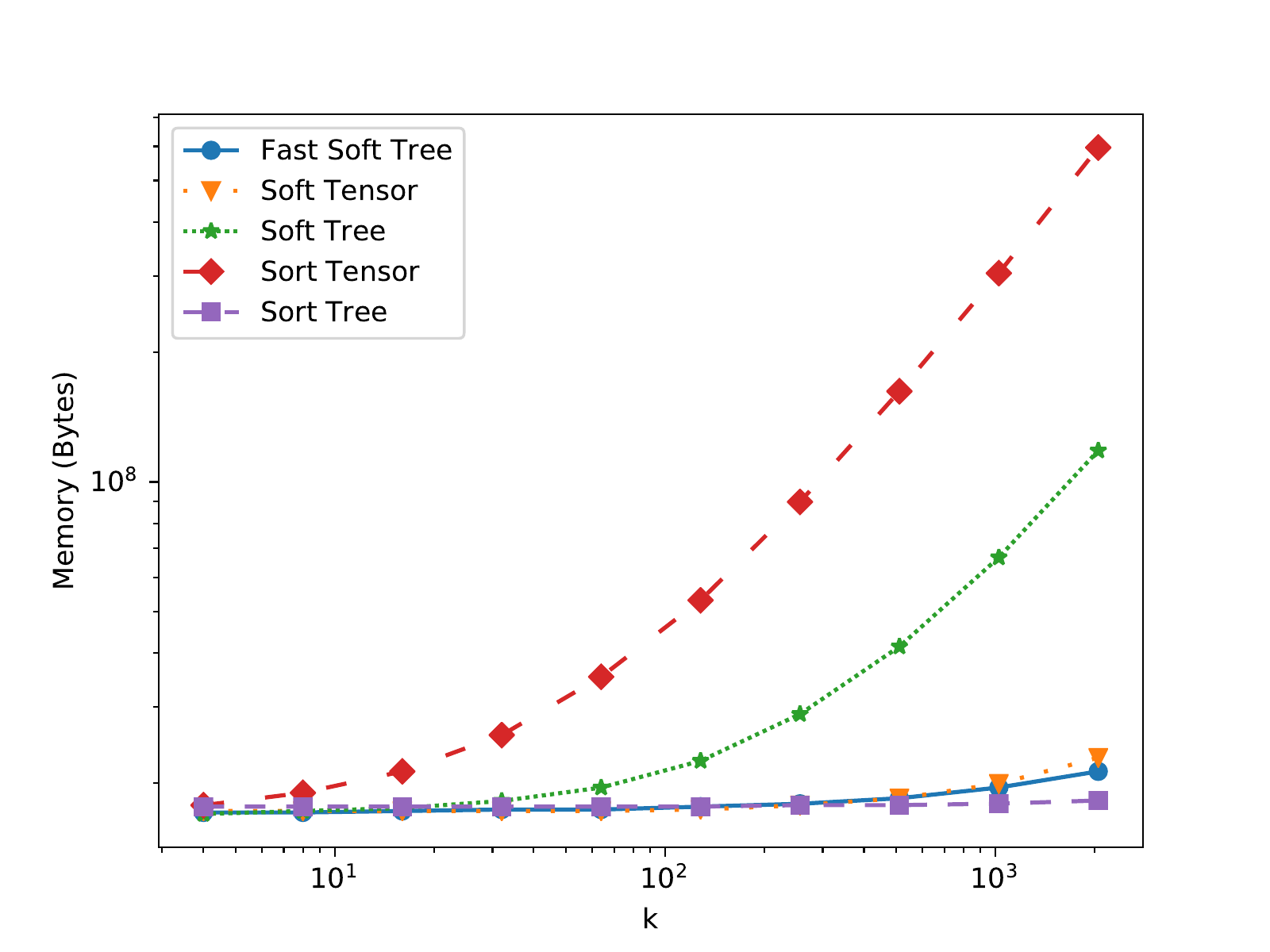}  \\
    \raisebox{5\normalbaselineskip}[0pt][0pt]{\rotatebox[origin=c]{90}{Vary m}} & \includegraphics[width=.5\linewidth]{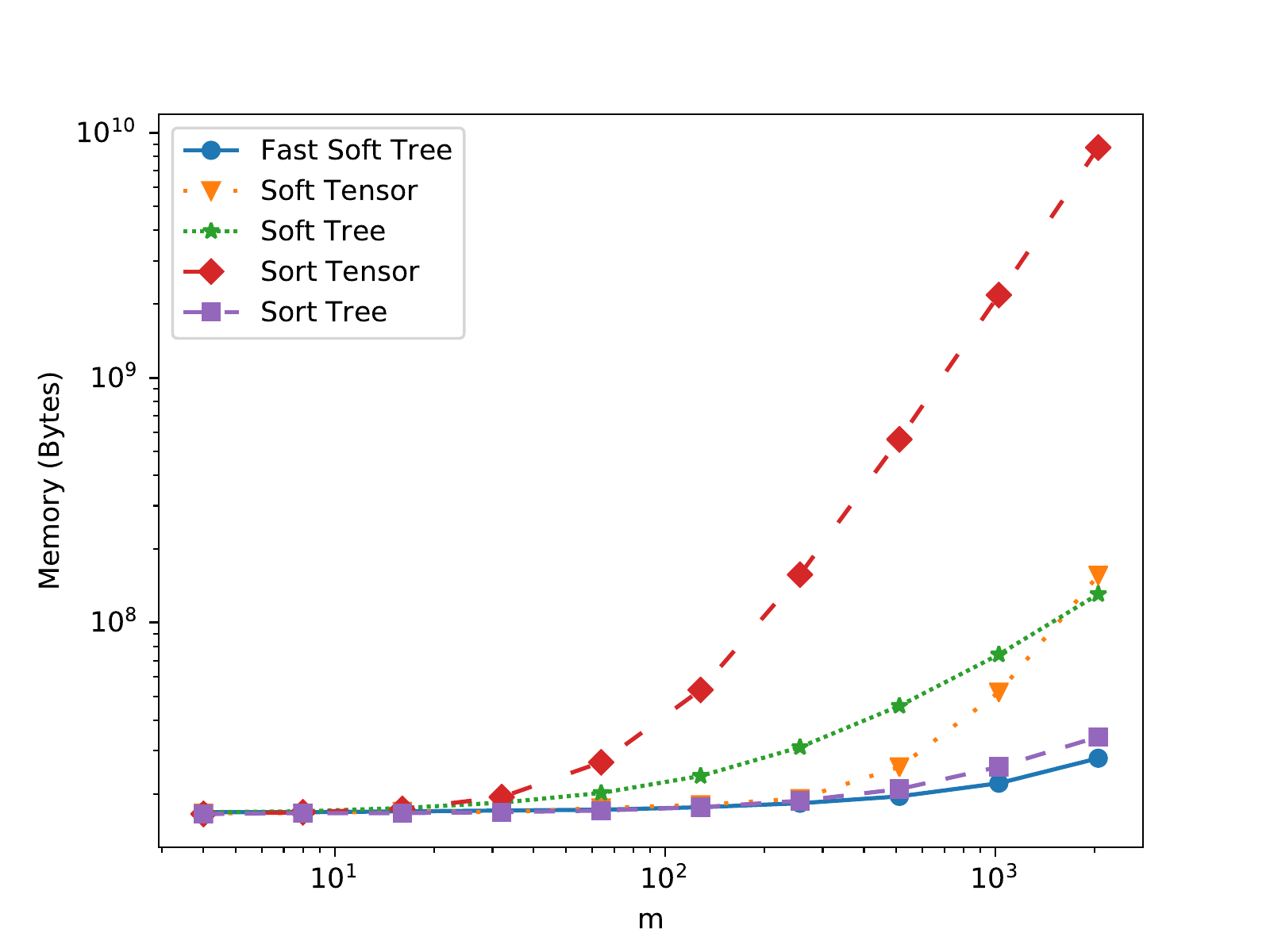} & \includegraphics[width=.5\linewidth]{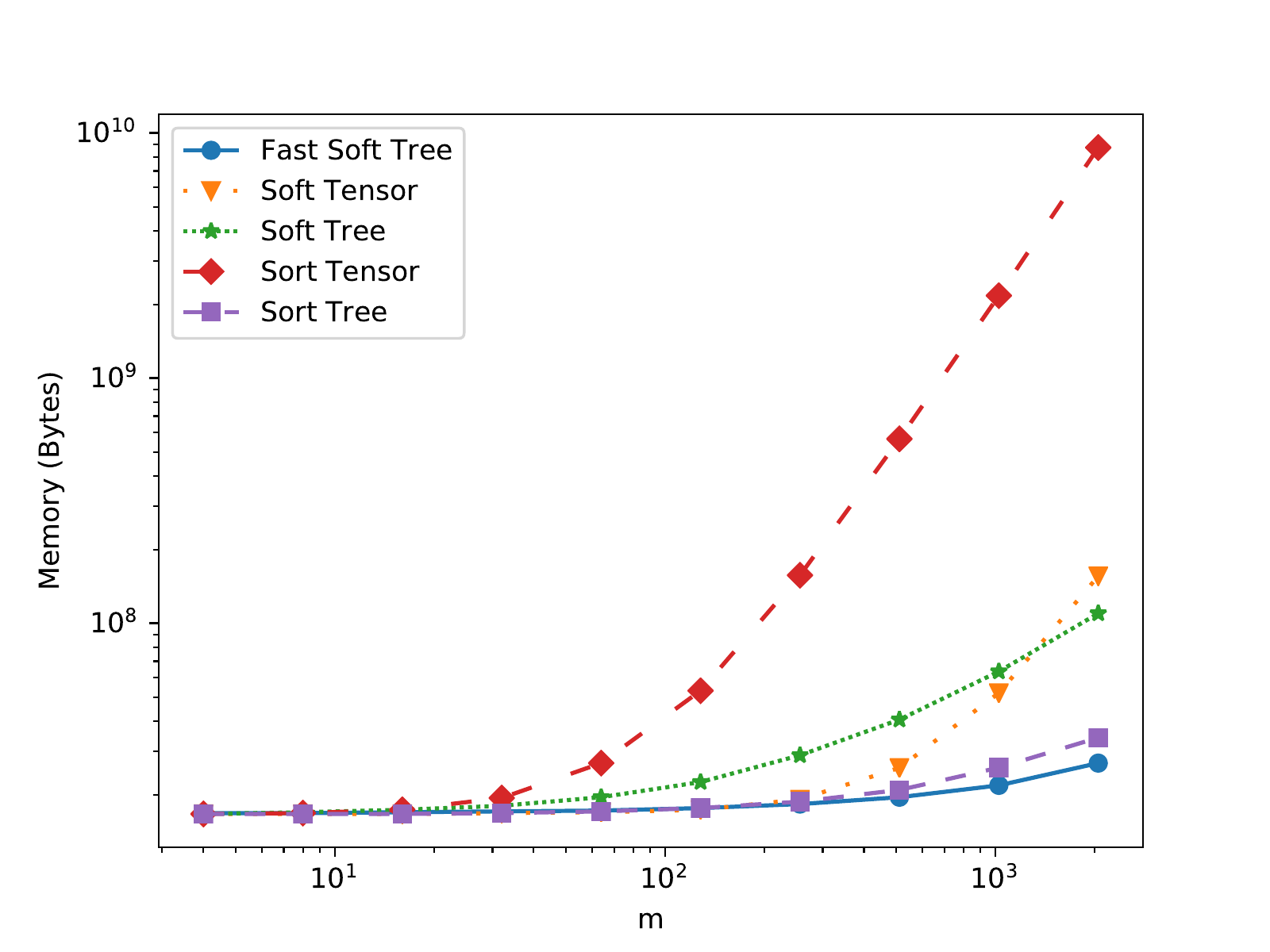}  \\
  \end{tabular}  
\caption{{\bf Memory usage.} Memory usage for each method is
  plotted. When $k$ is varied, it is $\in\{4, 8, 16, \ldots 256\}$
  with $n=32$ and $m=128$. When $m$ is varied, it is $\in\{4, 8, 16,
  \ldots 256\}$ with $n=32$ and $k=128$. We measure the average over
  twenty replicates. FastSoftTree was run with
  $\alpha=1.1$.
  \label{fig:memory}}
\end{table}

\subsection*{Propagation depth in SortTree}
The SortTree method does not always need to fetch a value from either
children ($A$ or $B$) in order to compute the new value in $A+B$. As a
result, the expected overall runtime may sometimes be $\in O(m\cdot n
+ k\log(k))$ rather than $\in O(m\cdot n + k\log(k)\log(m))$.

Table~\ref{table:vary-k-uni} and table~\ref{table:vary-k-exp}
demonstrate the in-practice propagation depth when $k$ is varied and
when the $X_i$ are uniformly and exponentially distributed,
respectively. Table~\ref{table:vary-m-uni} and
table~\ref{table:vary-m-exp} demonstrate the in-practice propagation
depth when $m$ is varied and when the $X_i$ are uniformly and
exponentially distributed, respectively. For both the uniformly and
exponentially distributed data, the average number of pops from each
leaf is sublinear with respect to $k$ and grows inversely proportional
to $m$.

\begin{table}
\centering
\scriptsize
\begin{tabular}{ r l l l l l l l }
$k$ &Root &  L1    & L2    & L3    & L4    & L5    & Leaves \\   
64   &64   & 15.50 & 6.000  & 3.750 & 3.312 & 3.031 & 2.953\\
128  &128  & 52.50 & 11.250 & 6.875 & 3.562 & 3.120  & 3.031\\
256  &256  & 34.00 & 12.50  & 5.750 & 3.687 & 3.250 & 3.031\\
512  &512  & 62.00 & 14.250 & 6.375 & 4.25 & 3.312  & 3.046\\
1024 &1024 & 76.50 & 13.250 & 6.125 & 3.937 & 3.218 & 2.984\\
2048 &2048 & 471.0 & 42.000 & 9.125 & 4.75 & 3.375  & 3.109\\
\end{tabular}
  \caption{Data is from varying $k$ with a uniform
    distribution and $m=64$, $n=1024$. This table shows the average number of times a node
    was asked to pop a value based on its depth in the SortTree. The
    left-most column represents the root, as the columns move to the
    right they traverse down the tree until they hit the root.}
\label{table:vary-k-uni}
\end{table}

\begin{table}
\centering
\scriptsize
\begin{tabular}{ r l l l l l l l }
$k$   & Root   &  L1    & L2    & L3    & L4    & L5    & Leaves \\ 
64   & 64.0   & 33.50  & 18.20 & 10.62 & 6.812 & 4.875 & 3.900 \\
128  & 128.0  & 65.50  & 34.25 & 18.62 & 10.75 & 6.843 & 4.890 \\
256  & 256.0  & 129.5  & 66.25  & 34.62 & 18.81 & 10.90 & 6.921\\
512  & 512.0  & 257.5  & 130.2 & 66.62 & 34.81 & 18.87 & 10.89 \\
1024 & 1024.0 & 513.5  & 258.2 & 130.6 & 66.68 & 34.84 & 18.82 \\
2048 & 2048.0 & 1025.0 & 502.7 & 252.9 & 66.18 & 34.59 & 18.70 \\
\end{tabular}
  \caption{Data is from varying $k$ with an exponential
    distribution and $m=64$, $n=1024$. This table shows the average number of times a node
    was asked to pop a value based on its depth in the SortTree. The
    left-most column represents the root, as the columns move to the
    right they traverse down the tree until they hit the root.}
\label{table:vary-k-exp}
\end{table}

\begin{table}
\centering
\scriptsize
\begin{tabular}{rllllllllllll}
m    & Root   &  L1    & L2    & L3    & L4    & L5    & L6    &  L7   & L8    & L9 &   L10    & L11 \\   
64   & 512.0 & 49.00 & 12.00 & 5.625 & 3.813 & 3.343 & 3.125 \\
128  & 512.0 & 48.50 & 10.75 & 5.500 & 3.813 & 3.281 & 3.078 & 2.961 \\ 
256  & 512.0 & 54.50 & 12.00 & 5.500 & 3.688 & 3.250 & 3.078 & 2.992  & 2.945 \\
512  & 512.0 & 71.00 & 14.50 & 6.000 & 4.125 & 3.281 & 3.109 & 2.977 & 2.945 & 2.928 \\
1024 & 512.0 & 45.50 & 14.25 & 5.875 & 4.063 & 3.313 & 3.094 & 2.961 & 2.938 & 2.930 & 2.913 \\
2048 & 512.0 & 46.50 & 10.75 & 4.625 & 3.500 & 3.156 & 3.047 & 2.961 & 2.930 & 2.928 & 2.921 & 2.913 \\
\end{tabular}
  \caption{Data is from varying $m$ with a uniform distribution and
  $k$ = 512 and $n$ = 1024. This table shows the average number of
  times a node was asked to pop a value based on its depth in the
  SortTree. The left-most column represents the root, as the columns
  move to the right they traverse down the tree until they hit the
  root. As $m$ grows so does the depth of the tree, so the right-most
  column which has entries  always represent the leaves. }
\label{table:vary-m-uni}
\end{table}

\begin{table}
\centering
\scriptsize
\begin{tabular}{rllllllllllll}
m    & Root   &  L1    & L2    & L3    & L4    & L5    & L6    &  L7   & L8    & L9  &  L10    & L11 \\   
64   & 512.0  & 257.5 & 130.3  & 66.63 & 34.81 & 18.91 & 10.93 \\
128  & 512.0  & 257.5 & 130.3  & 66.63 & 34.75 & 18.81 & 10.87 & 6.914 \\
256  & 512.0  & 257.5 & 130.2  & 66.62 & 34.81 & 18.81& 10.87 & 6.890& 4.906\\
512  & 512.0  & 257.5 & 130.3  & 66.63 & 34.75 & 18.84 & 10.87 & 6.906 & 4.914 & 3.914 \\
1024 & 512.0  & 257.5 & 130.3  & 66.63 & 34.81 & 18.88 & 10.90 & 6.914 & 4.902 & 3.919 & 3.418 \\
2048 & 512.0  & 46.50 & 10.75  & 4.630 & 3.500 & 3.156 & 3.046 & 2.960 & 2.929 & 2.927 & 2.920 & 2.912 \\
\end{tabular}
\caption{Data is from varying $m$ with an exponential distribution and
  $k$ = 512 and $n$ = 1024. This table shows the average number of
  times a node was asked to pop a value based on its depth in the
  SortTree. The left-most column represents the root, as the columns
  move to the right they traverse down the tree until they hit the
  root. As $m$ grows so does the depth of the tree, so the right-most
  column which has entries  always represent the leaves. }
\label{table:vary-m-exp}
\end{table}

\section*{Discussion}
Although the FastSoftTree method has the best theoretical runtime of
all methods here (\emph{e.g.}, when $\alpha=1.05$), the SortTree
method performed best in practice for these experiments. One reason
for the practical efficiency of the SortTree is the fact that it's
unlikely that every layer between the root and leaves are visited when
computing each next value in the sort of $X_1+X_2+\cdots +X_m$
(tables~\ref{table:vary-k-uni},~\ref{table:vary-k-exp},~\ref{table:vary-m-uni},~\&~\ref{table:vary-m-exp}).

Both sorting methods have a significant advantage in cache
performance. All data in SortTensor and SortTree are held in arrays
which are contiguous in memory. The soft heap requires pointers to
achieve its fast theoretical runtime which may not be able to take
advantage of the speed of a CPUs cache over RAM. If the data is
limited in size to the number of particles in the universe which is
approximately $2^{270}$\cite{eddington:mathematical}, the $\log(k)$
constant will be $\leq 270$. This may be small enough that, combined
with its cache performance, the $\log(k)$ term is is dwarfed by the
soft heap's runtime constant.

Finding the optimal $\alpha$ for any $n,m,k$ problem could also lead
to better performance in practice for FastSoftTree. It may even be
possible to mix different $\alpha$ values throughout layers of the
tree. Allowing $\alpha \geq 2$ at the roots gives a faster LOHify time
while a smaller $\alpha$ higher up in the tree means less values
generated overall.

If it were possible to efficiently implement the pairwise $A+B$
selection without soft heaps (or reminiscent pointer basic data
structure), then the practical performance of the FastSoftTree would
likely be much better. Layer-ordered heaps are essentially vectors
iteratively partitioned (as from quicksort); they can be stored in a
contiguous manner and with low overhead overall (including a fast
implementation of median-of-medians instead of soft-select for the
one-dimensional selection).

Although they are quite simple, layer-ordered heaps can be used to
describe other problems: For instance, given a fast, in-place method
to construct a layer-ordered heap of an array, it is possible to
implement an efficient one-dimensional selection. This is accomplished
as a simplified form of the algorithm from
theorem~\ref{thm:concatenation-selection-runtime}. A layer-ordered
heap is built on the full array, and then layers are added (small
layers first) until the selection threshold is exceeded. At this
point, the layer of the layer-ordered heap that resulted in exceeding
the selection threshold is itself partitioned into a layer-ordered
heap, and selection on the remaining values is performed only within
that layer. This is continued until the exact selection threshold is
met. Assuming a linear-time layer-ordered heapification algorithm, the
runtime of this one-dimensional $k$-selection will be characterized by
the recurrence $r(k)=k + r((\alpha-1)\cdot k)$; the final layer will
be added only if we have $<k$ so far, and we will overshoot by at most
a factor of $\alpha$ (and thus the final layer has size $\alpha\cdot k
- k$. For any $\alpha<2$ and not close to 1, this achieves linear-time
one-dimensional selection by essentially achieving a more and more
sorted ordering at the selection threshold.

It is likely that algorithms like these using layer-ordered heaps will
be useful on other problems which currently use sorting when it is not
completely necessary such as non-parametric test statistics.

\section*{Acknowledgements}
This work was supported by grant number 1845465 from the National
Science Foundation.

\section*{Supplemental information}
\subsection*{Code availability}
Both {\tt C++} and Python source code for all methods is available at
\url{https://bitbucket.org/seranglab/cartesian-selection-with-layer-ordered-heap}
(MIT license, free for both academic and commercial use).

\end{document}